\documentclass[11pt]{article}%
\usepackage{amsfonts}
\usepackage{amsmath}
\usepackage{amssymb}
\usepackage{fullpage}
\usepackage[colorlinks=true,linkcolor=blue,citecolor=blue,plainpages=false,pdfpagelabels]%
{hyperref}
\usepackage{graphicx}%
\providecommand{\U}[1]{\protect\rule{.1in}{.1in}}
\newtheorem{theorem}{Theorem}

\newtheorem{definition}[theorem]{Definition}

\newtheorem{notation}[theorem]{Notation}

\newtheorem{proposition}[theorem]{Proposition}
\newtheorem{remark}[theorem]{Remark}

\newenvironment{proof}[1][Proof]{\noindent\textbf{#1.} }{\ \rule{0.5em}{0.5em}}
\numberwithin{equation}{section}
\begin{document}

\title{Relative entropy of steering:\ On its definition and properties}
\author{Eneet Kaur\thanks{ Hearne Institute for Theoretical Physics, Department of
Physics and Astronomy, Louisiana State University, Baton Rouge, Louisiana
70803, USA}
\and Mark M. Wilde\footnotemark[1] \thanks{Center for Computation and Technology,
Louisiana State University, Baton Rouge, Louisiana 70803, USA}}
\date{\today}
\maketitle

\begin{abstract}
In [Gallego and Aolita, \textit{Physical Review X} \textbf{5}, 041008 (2015)],
the authors proposed a definition for the relative entropy of steering and
showed that the resulting quantity is a convex steering monotone. Here we
advocate for a different definition for relative entropy of steering, based on
well grounded concerns coming from quantum Shannon theory. We prove that this
modified relative entropy of steering is a convex steering monotone.
Furthermore, we establish that it is uniformly continuous and faithful, in
both cases giving quantitative bounds that should be useful in applications.
We also consider a restricted relative entropy of steering which is relevant
for the case in which the free operations in the resource theory of steering
have a more restricted form (the restricted operations could be more relevant
in practical scenarios). The restricted relative entropy of steering is
convex, monotone with respect to these restricted operations, uniformly
continuous, and faithful.

\end{abstract}

\section{Introduction}

Quantum steering corresponds to the scenario in which two parties, typically
called Alice and Bob, share a quantum state, and Alice can have an effect on
the state of Bob's system if she performs local measurements on hers
\cite{Einstein1935,Schroedinger1935,Wiseman2006,CS16}. For certain quantum
states, this effect cannot be explained in a classical way, and such states
are said to be steerable \cite{Wiseman2006}. Steerable states are necessarily
entangled but do not necessarily violate a Bell inequality \cite{CS16}.

Quantum steering is relevant as a resource in the context of one-sided
device-independent quantum key distribution \cite{Branciard2012}, in which the
goal is to distill secret key between Alice, who does not trust the quantum
device provided to her, and Bob, who trusts his quantum device. Motivated by
this, the authors of \cite{Gallego2015}\ developed a resource theory of
quantum steering, establishing free states in the resource theory as the
unsteerable ones and the free operations as one-way local operations and
classical communication (1W-LOCC), which preserve the free states. The same
authors also defined a steering monotone to be a function that does not
increase on average under 1W-LOCC, they proposed a definition for the relative
entropy of steering, and they proved that their proposed quantity is a
steering monotone.

The relative entropy of steering proposed in \cite{Gallego2015} can be
considered in a game-theoretic context with two players and the pay-off
function given by the quantum relative entropy. The relative entropy of
steering is a function of an assemblage $\{\hat{\rho}_{B}^{a,x}\}_{a,x}$,
defined to be the set of unnormalized states that result on Bob's system $B$
after Alice performs measurement $x\in\mathcal{X}$ and receives outcome
$a\in\mathcal{A}$, where $\mathcal{X}$ and $\mathcal{A}$ are finite alphabets.
That is, if Alice and Bob share the state $\rho_{AB}$ and Alice performs a
positive operator-valued measure (POVM)\ $\{\Lambda_{a}^{(x)}\}_{a}$ on her
system, where $\Lambda_{a}^{(x)}\geq0$ and $\sum_{a}\Lambda_{a}^{(x)}=I_{A}$,
then the resulting assemblage would be $\{\hat{\rho}_{B}^{a,x}%
=\operatorname{Tr}_{A}([\Lambda_{a}^{(x)}\otimes I_{B}]\rho_{AB})\}_{a,x}$.
The relative entropy of steering proposed in \cite{Gallego2015} quantifies how
distinguishable a given assemblage is from one that has a classical
description, in terms of the quantum relative entropy \cite{U62}. In
particular, let us say that Player~1's goal is to maximize the quantum
relative entropy between the two assemblages, and he is allowed to perform any
1W-LOCC\ operation in order to do so. Player~2's goal is to minimize the
quantum relative entropy by picking an assemblage that has a classical
description. Clearly, we have to pick an order in which the players take their
turns. In \cite{Gallego2015}, the authors had Player 2 go first, and then
Player~1 next. This means that Player~1 can react to the strategy of Player~2,
and in particular that the quantity in mathematical terms looks like (crudely)%
\begin{equation}
\inf_{\operatorname{LHS}}\sup_{\operatorname{1WLOCC}}D, \label{eq:old-rel-ent}%
\end{equation}
where $\operatorname{LHS}$ is the set of assemblages having a classical
description, $\operatorname{1WLOCC}$ is the set of 1W-LOCC\ operations, and
$D$ is the quantum relative entropy payoff function (we will define all of
this in much more detail later).

The main purpose of the present paper is to advocate for a different
definition of the relative entropy of steering in which the order of play
described above is exchanged, so that (crudely), the quantity we are proposing
is%
\begin{equation}
\sup_{\operatorname{1WLOCC}}\inf_{\operatorname{LHS}}D. \label{eq:new-rel-ent}%
\end{equation}
The interpretation is thus that Player~1 first acts to maximize $D$ by
\textquotedblleft playing\textquotedblright\ a 1W-LOCC\ operation, to which
Player~2 can react by \textquotedblleft playing\textquotedblright\ an
assemblage having a classical description. Our alternate definition for the
relative entropy of steering might seem like a minor modification, but we
offer three compelling reasons for our proposal:

\begin{enumerate}
\item The optimization order for the quantity in \eqref{eq:new-rel-ent}\ is
consistent with all previously known information-theoretic measures of dynamic
resources as considered in quantum Shannon theory \cite{W15book}, including
Holevo information of a channel \cite{H06}, mutual information of a channel
\cite{PhysRevA.56.3470,ieee2002bennett}, coherent information of a channel
\cite{PhysRevA.54.2629}, squashed entanglement of a channel \cite{TGW14IEEE},
Rains information of a channel \cite{TWW14}, etc.

\item The quantity in \eqref{eq:new-rel-ent}\ is never larger than that in
\eqref{eq:old-rel-ent}\ (due to the order of optimizations), and given that
the main application of relative entropic quantifiers in quantum Shannon
theory has been to get tight upper bounds on distillable entanglement or
secret key \cite{R01,HHHO05,HHHO09,TWW14,TBR15,PLOB15,WTB16}, we suspect that
the quantity in \eqref{eq:new-rel-ent}\ will be the right one to use in future applications.

\item The game-theoretic interpretation from \cite{vDGG05}\ would say that
\eqref{eq:new-rel-ent} quantifies the statistical strength of Player~1 to
convince Player~2 that the underlying assemblage demonstrates steering, and thus
represents a stronger measure or proof of the statistical strength of
steerability than does~\eqref{eq:old-rel-ent}.
\end{enumerate}

\noindent We elaborate more on the first point in
Section~\ref{sec:further-just}.

In the remainder of the paper, we review some preliminaries in
Section~\ref{sec:prelim}\ and provide a formal definition for our proposed
relative entropy of steering in Section~\ref{sec:rel-ent-steer}\ (we refer to
this quantity simply as \textquotedblleft the relative entropy of
steering\textquotedblright\ in the remainder of the paper). In
Section~\ref{sec:monotone}, we prove that the relative entropy of steering is
a steering monotone, and in Section~\ref{sec:convexity}\ we prove that it is a
convex function of the assemblage for which it is evaluated. Thus, the
relative entropy of steering is a convex steering monotone according to
\cite[Definition~2]{Gallego2015}. Section~\ref{sec:upper-bounds}\ establishes
upper bounds on the relative entropy of steering. In Section~\ref{sec:cont},
we define a metric for assemblages (\textquotedblleft trace distance of
assemblages\textquotedblright), and we prove that the relative entropy of
steering is uniformly continuous with respect to this metric (we give
quantitative continuity bounds). In Section~\ref{sec:faithful}, we prove that
the relative entropy of steering is faithful, and we give quantitative
faithfulness bounds.

As discussed in \cite{KWW16}, we can consider a restricted class of
1W-LOCC\ operations that might have more relevance in practical scenarios, in
which classical communication from Bob to Alice reaches Alice only after she
obtains the output of her black box. With this in mind, we define a restricted
relative entropy of steering, and we prove that it is a restricted steering
monotone, faithful, and uniformly continuous with respect to a metric relevant
for restricted 1W-LOCC.

\section{Preliminaries}

\label{sec:prelim}In the introduction, we discussed assemblages as arising
from a local measurement of Alice on a bipartite state that she shares with
Bob. However, the common approach in the steering literature \cite{CS16}, also
known as the one-sided device-independent approach, is to consider an
assemblage on its own, being defined as a set $\{\hat{\rho}_{B}^{a,x}\}_{a,x}$
of arbitrary positive semi-definite operators constrained by the no-signaling
principle. From the one-sided device-independent perspective, we think of
Alice's system as being a black box, taking a classical input $x\in
\mathcal{X}$ and producing a classical ouptut $a\in\mathcal{A}$, where
$\mathcal{X}$ and $\mathcal{A}$ are finite alphabets. The no-signaling
principle is that the reduced state of Bob's system should not depend on the
input $x$ to Alice's black box if the output $a$ is not available to him:%
\begin{equation}
\sum_{a}\hat{\rho}_{B}^{a,x}=\sum_{a}\hat{\rho}_{B}^{a,x^{\prime}}\quad\forall
x,x^{\prime}\in\mathcal{X}.\label{eq:no-sig-constr}%
\end{equation}
We can then define $\rho_{B}:=\sum_{a}\hat{\rho}_{B}^{a,x}$ and the last
constraint on an assemblage is that $\rho_{B}$ is a quantum state. With this
last constraint, we see that $\operatorname{Tr}(\hat{\rho}_{B}^{a,x})$ can be
interpreted as a conditional probability distribution $p_{\overline{A}|X}$, so that
$p_{\overline{A}|X}(a|x)=\operatorname{Tr}(\hat{\rho}_{B}^{a,x})$.

As discussed in \cite{KWW16}, one can think of an assemblage as being similar
to a quantum broadcast channel \cite{YHD2006}, accepting a classical input $x$
from a sender and producing a classical output $a$ with probability
$\operatorname{Tr}(\hat{\rho}_{B}^{a,x})$ for one receiver and a quantum
output $\hat{\rho}_{B}^{a,x}/\operatorname{Tr}(\hat{\rho}_{B}^{a,x})$ for the
other receiver if $\operatorname{Tr}(\hat{\rho}_{B}^{a,x})\neq0$. However,
this perspective is not fully complete, given that the quantum system $B$ is
accessible to Bob before the input $x$ is chosen. In any case, we say that an
assemblage is a \textit{dynamic resource} in the sense of \cite{DHW05RI}, in
that its behavior is modified depending on the input $x$.

An assemblage \textit{does not demonstrate steering} if arises from a classical, shared
random variable $\Lambda$ in the following sense~\cite{Wiseman2006}:
\begin{equation}
\hat{\rho}_{B}^{a,x}=\sum_{\lambda}p_{\Lambda}(\lambda)\ p_{\overline{A}|X\Lambda
}(a|x,\lambda)\ \rho_{B}^{\lambda},
\end{equation}
where $p_{\Lambda}(\lambda)$ is a probability distribution for $\Lambda$,
$p_{\overline{A}|X\Lambda}$ is a conditional probability distribution, and $\rho
_{B}^{\lambda}$ is a quantum state. The above structure indicates that the
correlations observed can be explained by a classical random variable
$\Lambda$, a copy of which is sent to both Alice and Bob, who then take
actions conditioned on a particular realization $\lambda$ of $\Lambda$. The
set of all assemblages that do not demonstrate steering is referred to as $\operatorname{LHS}$
(short for assemblages having a \textquotedblleft local-hidden-state
model\textquotedblright).

As discussed in the introduction, the most general free operations allowed in
the context of quantum steering are 1W-LOCC \cite{Gallego2015,KWW16}. As a
particular example, starting with a given assemblage $\{\hat{\rho}_{B}%
^{a,x}\}_{a,x}$, it is possible for Bob to perform a generalized measurement
on his system, specified as the following measurement channel acting on an
input state $\sigma_{B}$:
\begin{equation}
\mathcal{M}_{B\rightarrow B^{\prime}Y}(\sigma_{B}):=\sum_{y}\mathcal{K}%
_{y}(\sigma_{B})\otimes|y\rangle\langle y|_{Y},
\end{equation}
where each $\mathcal{K}_{y}$ is a completely positive trace-non-increasing
map, such that the sum map $\sum_{y}\mathcal{K}_{y}$ is trace preserving. Note
that each map $\mathcal{K}_{y}$ can be written as $\mathcal{K}_{y}(\sigma
_{B})=\sum_{t}K_{y,t}\sigma_{B}K_{y,t}^{\dag}$, such that $\sum_{y,t}%
K_{y,t}^{\dag}K_{y,t}=I$ and where each $K_{y,t}$ is a Kraus operator taking a
vector in $\mathcal{H}_{B}$ to a vector in $\mathcal{H}_{B^{\prime}}$. Also,
$\{|y\rangle\}_{y}$ denotes an orthonormal basis. Bob can then communicate the
classical result $y$ to Alice, who chooses the input $x$ to her black box
according to a classical channel $p_{X|Y}(x|y)$. The state after these
operations is
\begin{equation}
\rho_{X\overline{A}B^{\prime}Y}:=\sum_{a,x,y}p_{X|Y}(x|y)|x\rangle\langle
x|_{X}\otimes|a\rangle\langle a|_{\overline{A}}\otimes\mathcal{K}_{y}(\hat{\rho
}_{B}^{a,x})\otimes|y\rangle\langle y|_{Y},\label{eq:state-after-1W-LOCC}%
\end{equation}
where $\{|x\rangle\}_{x}$ and $\{|a\rangle\}_{a}$ denote orthonormal bases.

We now recall the defintion of quantum relative entropy, one of the main tools
used in this paper. The quantum relative entropy $D(\rho\Vert\sigma)$ accepts
two quantum states $\rho$ and $\sigma$ as input and outputs a non-negative
real number. It is defined as \cite{U62}%
\begin{equation}
D(\rho\Vert\sigma):=\operatorname{Tr}(\rho\left[  \log_{2}\rho-\log_{2}%
\sigma\right]  )
\end{equation}
if the support of $\rho$ is contained in the support of $\sigma$ and otherwise
it is set to $+\infty$. In the above definition, we take the common convention
that the operator logarithms are defined on the support of their arguments.
The most critical property of quantum relative entropy is that it is monotone
with respect to a quantum channel $\mathcal{N}$ \cite{Lindblad1975,U77}, in
the sense that%
\begin{equation}
D(\rho\Vert\sigma)\geq D(\mathcal{N}(\rho)\Vert\mathcal{N}(\sigma)).
\end{equation}
The quantum relative entropy obeys the following property:%
\begin{equation}
D\!\left(  \sum_{x}r(x)|x\rangle\langle x|\otimes\lambda^{x}\middle\Vert
\sum_{x}s(x)|x\rangle\langle x|\otimes\mu^{x}\right)  =\sum_{x}r(x)D(\lambda
^{x}\Vert\mu^{x})+D(r\Vert s), \label{eq:rel-ent-block-prop}%
\end{equation}
which holds for probability distributions $r$ and $s$, sets of density
operators $\{\lambda^{x}\}_{x}$ and $\{\mu^{x}\}_{x}$, and an orthonormal basis $\{|x\rangle
\}_{x}$. Note that if we write $D(p\Vert q)$ for
probability distributions $p$ and $q$, then it is implicit that these
distributions are encoded along the diagonal of a density operator, so that
the corresponding states are commuting.

The quantum entropy is defined as $H(G)_{\kappa}:= H(\kappa_{G}%
):=-\operatorname{Tr}(\kappa_{G}\log_{2}\kappa_{G})$ for a state $\kappa_{G}$
on system $G$.

\section{Relative entropy of steering}

\label{sec:rel-ent-steer}In this section, we first give our proposed
definition of relative entropy of steering. We then show that it is a convex
steering monotone. The subsections thereafter establish upper bounds on it,
the trace distance of assemblages as a metric on assemblages, uniform
continuity of the relative entropy of steering, and its faithfulness.

\begin{definition}
[Relative entropy of steering]\label{def:rel-ent-steer}Let $\{\hat{\rho}%
_{B}^{a,x}\}_{a,x}$ denote an assemblage. We define the relative entropy of
steering as follows:%
\begin{equation}
R_{S}(\overline{A};B)_{\hat{\rho}}:=\sup_{\{p_{X|Y},\{\mathcal{K}_{y}\}_{y}\}}%
\inf_{\{\hat{\sigma}_{B}^{a,x}\}_{a,x}\in\operatorname{LHS}}D(\rho_{X\overline
{A}B^{\prime}Y}\Vert\sigma_{X\overline{A}B^{\prime}Y}),
\end{equation}
where%
\begin{align}
\rho_{X\overline{A}B^{\prime}Y} &  :=\sum_{x,a,y}p_{X|Y}(x|y)|x\rangle\langle
x|_{X}\otimes|a\rangle\langle a|_{\overline{A}}\otimes\mathcal{K}_{y}(\hat{\rho
}_{B}^{a,x})\otimes|y\rangle\langle y|_{Y},\\
\sigma_{X\overline{A}B^{\prime}Y} &  :=\sum_{x,a,y}p_{X|Y}(x|y)|x\rangle\langle
x|_{X}\otimes|a\rangle\langle a|_{\overline{A}}\otimes\mathcal{K}_{y}(\hat{\sigma
}_{B}^{a,x})\otimes|y\rangle\langle y|_{Y},
\end{align}
$\operatorname{LHS}$ denotes the set of all assemblages having a
local-hidden-state model, and $\{p_{X|Y},\{\mathcal{K}_{y}\}_{y}\}$ denotes a
1W-LOCC\ operation as described in \eqref{eq:state-after-1W-LOCC}.
\end{definition}

\begin{remark}
By using the property of relative entropy recalled in
\eqref{eq:rel-ent-block-prop}, the definition of relative entropy of steering
given in \cite{Gallego2015} can be written as%
\begin{equation}
\inf_{\{\hat{\sigma}_{B}^{a,x}\}_{a,x}\in\operatorname{LHS}}\sup
_{\{p_{X|Y},\{\mathcal{K}_{y}\}_{y}\}}D(\rho_{X\overline{A}B^{\prime}Y}\Vert
\sigma_{X\overline{A}B^{\prime}Y}),
\end{equation}
with the symbols involved defined as above.
\end{remark}

\subsection{Justification for Definition~\ref{def:rel-ent-steer}}

\label{sec:further-just}We gave three reasons in the introduction that
advocate for Definition~\ref{def:rel-ent-steer}\ to be the relative entropy of
steering over the definition given in \cite{Gallego2015}. We now elaborate on
the first reason, which is that the order of optimizations in
Definition~\ref{def:rel-ent-steer} is consistent with the order of
optimizations given in all known information-theoretic measures of a dynamic
quantum resource. Since an assemblage is a dynamic resource as discussed in
Section~\ref{sec:prelim}, we see no strong reason why the order of
optimizations in the relative entropy of steering should not be consistent
with all of these other measures.

We first briefly recall some definitions. The quantum mutual information and
coherent information of a bipartite state $\rho_{AB}$ can be defined,
respectively, as%
\begin{align}
I(A;B)_{\rho}  &  :=\inf_{\sigma_{B}}D(\rho_{AB}\Vert\rho_{A}\otimes\sigma
_{B}),\\
I(A\rangle B)_{\rho}  &  :=\inf_{\sigma_{B}}D(\rho_{AB}\Vert I_{A}%
\otimes\sigma_{B}),
\end{align}
where the optimizations are with respect to a quantum state $\sigma_{B}$ (see,
e.g., \cite[Section~11.8.1]{W15book}). The conditional mutual information of a
tripartite state $\rho_{ABE}$ can be defined as%
\begin{equation}
I(A;B|E)_{\rho}:=I(A;BE)_{\rho}-I(A;E)_{\rho}.
\end{equation}

A dynamic resource of primary interest in quantum Shannon theory is a quantum
channel $\mathcal{N}_{A\rightarrow B}$, which accepts a state on an input
quantum system $A$ and physically transforms it to a state on an output
quantum system $B$. One of the main goals of quantum Shannon theory is to
determine capacities of a quantum channel for various communication tasks. The
result of many years of effort is that different functions of a quantum
channel characterize its different capacities. For example, the classical
capacity is characterized by the Holevo information
\cite{Hol98,PhysRevA.56.131,H06}, the entanglement-assisted capacity by the
channel's mutual information \cite{PhysRevLett.83.3081,ieee2002bennett}, and
the quantum capacity by the channel's coherent information
\cite{PhysRevA.55.1613,capacity2002shor,ieee2005dev}, respectively defined as%
\begin{align}
\sup_{\rho_{XA}}I(X;B)_{\mathcal{N}(\rho)}  &  =\sup_{\rho_{XA}}\inf
_{\sigma_{B}}D(\mathcal{N}_{A\rightarrow B}(\rho_{XA})\Vert\rho_{X}%
\otimes\sigma_{B}),\\
\sup_{\rho_{RA}}I(R;B)_{\mathcal{N}(\rho)}  &  =\sup_{\rho_{RA}}\inf
_{\sigma_{B}}D(\mathcal{N}_{A\rightarrow B}(\rho_{RA})\Vert\rho_{R}%
\otimes\sigma_{B}),\\
\sup_{\rho_{RA}}I(R\rangle B)_{\mathcal{N}(\rho)}  &  =\sup_{\rho_{RA}}%
\inf_{\sigma_{B}}D(\mathcal{N}_{A\rightarrow B}(\rho_{RA})\Vert I_{R}%
\otimes\sigma_{B}).
\end{align}
In the first line, there is a constraint that system $X$ is a classical system
while system $A$ is quantum. In the last two expressions, systems $R$ and $A$
are quantum. The expressions on the right-hand side indicate that the
information quantities can be thought of as a comparison between the output of
the actual channel and the output of a useless channel, which is one that
traces out the input system $A$ and replaces it with the state $\sigma_{B}$.
We see in each case that the order of optimization is critically taken to be
such that the maximizing player goes first, inputting a state intended to give
the best possible discrimination between the channel $\mathcal{N}%
_{A\rightarrow B}$ of interest and a useless channel. The minimizing player
goes second, being able to react to the play of the maximizer by choosing the
worst possible useless channel depending on the state $\mathcal{N}%
_{A\rightarrow B}(\rho_{XA})$ or $\mathcal{N}_{A\rightarrow B}(\rho_{RA})$.

Other information measures that have been used to give upper bounds on
communication tasks include the squashed entanglement of a channel
\cite{TGW14IEEE}, the Rains information of a channel \cite{TWW14}, and a
channel's relative entropy of entanglement \cite{TWW14,PLOB15,WTB16}. These
are defined respectively as%
\begin{align}
&  \sup_{\psi_{RA}}\inf_{\mathcal{S}_{E\rightarrow E^{\prime}}}I(A;B|E^{\prime
})_{\omega},\\
&  \sup_{\rho_{RA}}\inf_{\sigma_{AB}\in\operatorname{PPT}^{\prime}%
}D(\mathcal{N}_{A\rightarrow B}(\rho_{RA})\Vert\sigma_{AB}),\\
&  \sup_{\rho_{RA}}\inf_{\sigma_{AB}\in\operatorname{SEP}}D(\mathcal{N}%
_{A\rightarrow B}(\rho_{RA})\Vert\sigma_{AB}),
\end{align}
where in the squashed entanglement of a channel, we take $\omega_{ABE^{\prime
}}:=\mathcal{S}_{E\rightarrow E^{\prime}}(\mathcal{U}_{A\rightarrow
BE}^{\mathcal{N}}(\psi_{RA}))$, with $\psi_{RA}$ a pure state, $\mathcal{U}%
_{A\rightarrow BE}^{\mathcal{N}}$ a fixed isometric extension of the channel
$\mathcal{N}_{A\rightarrow B}$, and $\mathcal{S}_{E\rightarrow E^{\prime}}$ a
channel known as a squashing channel. In the latter two lines,
$\operatorname{PPT}^{\prime}$ is a set of subnormalized states related to and
containing the positive-partial-transpose (PPT)\ states, and
$\operatorname{SEP}$ denotes the set of separable, unentangled states. Thus,
the interpretation is the same as above: an input to the channel is chosen and
then an adversary reacts to this input by trying to minimize the
discrimination measure. Note that the latter two quantities have found
application as upper bounds on quantum capacity and private capacity, in part
because they involve a comparison with a PPT\ state, which is useless for
quantum communication \cite{R01}, and with a separable state, which is useless
for private communication \cite{CLL04,HHHO05,HHHO09}.

Thus, given the above list of information measures which have found extensive
use throughout quantum Shannon theory and given that each of them have the
optimization order as $\sup\inf$, we suspect that this optimization order will
be the right approach to take for the relative entropy of steering. Note also
that, similar to all of the above information measures, the relative entropy
of steering involves a comparison between a given assemblage and another which
is useless in the context of steering, in the sense that the latter has a
local-hidden-state model and thus does not demonstrate steering.

\subsection{Steering monotone\label{sec:monotone}}

We now prove that the relative entropy of steering is a steering monotone,
however deferring the faithfulness proof until Section~\ref{sec:faithful}:

\begin{theorem}
[Steering monotone]Let $\{\hat{\rho}_{B}^{a,x}\}_{a,x}$ be an assemblage, and
suppose that%
\begin{equation}
\left\{  \hat{\rho}_{B_{f},z}^{a_{f},x_{f}}:=\sum_{a,x}p(a_{f}|x_{f}%
,x,a,z)p(x|x_{f},z)\mathcal{K}_{z}(\hat{\rho}_{B}^{a,x})/p(z)\right\}
_{a_{f},x_{f}},
\end{equation}
is an assemblage that arises from it by the action of a general
1W-LOCC\ operation (see \cite[Definition~1]{Gallego2015} and \cite{KWW16}),
where%
\begin{equation}
p(z):=\operatorname{Tr}\!\left(  \mathcal{K}_{z}\!\left(  \sum_{a}\hat{\rho
}_{B}^{a,x}\right)  \right)  =\operatorname{Tr}(\mathcal{K}_{z}(\rho_{B})).
\end{equation}
Then%
\begin{equation}
\sum_{z}p(z)R_{S}(\overline{A}_{f};B_{f})_{\hat{\rho}_{z}}\leq R_{S}(\overline
{A};B)_{\hat{\rho}}.\label{eq:1W-LOCC-monotone}%
\end{equation}

\end{theorem}

\begin{proof}
Let $\{\hat{\sigma}_{B}^{a,x}\}_{a,x}$ be an LHS\ assemblage, and suppose that%
\begin{equation}
\left\{  \hat{\sigma}_{B_{f},z}^{a_{f},x_{f}}:=\sum_{a,x}p(a_{f}%
|x_{f},x,a,z)p(x|x_{f},z)\mathcal{K}_{z}(\hat{\sigma}_{B}^{a,x})/q(z)\right\}
_{a_{f},x_{f}},
\end{equation}
is an LHS\ assemblage that arises from it by the action of the same
1W-LOCC\ operation as above, where%
\begin{equation}
q(z):=\operatorname{Tr}\!\left(  \mathcal{K}_{z}\!\left(  \sum_{a}\hat{\sigma
}_{B}^{a,x}\right)  \right)  =\operatorname{Tr}(\mathcal{K}_{z}(\sigma_{B})).
\end{equation}
The assemblage $\{\hat{\sigma}_{B_{f},z}^{a_{f},x_{f}}\}_{a,x}$ is guaranteed
to be an LHS\ assemblage by \cite[Theorem~1]{Gallego2015}. Consider that, in
accordance with the definition of $R_{S}(\overline{A}_{f};B_{f})_{\hat{\rho}_{z}}$,
the assemblages $\{\hat{\rho}_{B_{f},z}^{a_{f},x_{f}}\}_{a_{f},x_{f}}$ and
$\{\hat{\sigma}_{B_{f},z}^{a_{f},x_{f}}\}_{a_{f},x_{f}}$ can be further
preprocessed by a $z$-dependent 1W-LOCC $\{p_{X_{f}|YZ=z},\{\mathcal{L}%
_{y}^{(z)}\}_{y}\}$, resulting in the following states:%
\begin{align}
\omega_{X_{f}\overline{A}_{f}B_{f}^{\prime}Y}^{z} &  :=\sum_{a_{f},x_{f},y}%
p(x_{f}|y,z)[x_{f}]\otimes\lbrack a_{f}]\otimes\mathcal{L}_{y}^{(z)}(\hat{\rho
}_{B_{f},z}^{a_{f},x_{f}})\otimes\lbrack y],\label{eq:smaller-state-def}\\
\tau_{X_{f}\overline{A}_{f}B_{f}^{\prime}Y}^{z} &  :=\sum_{a_{f},x_{f},y}%
p(x_{f}|y,z)[x_{f}]\otimes\lbrack a_{f}]\otimes\mathcal{L}_{y}^{(z)}%
(\hat{\sigma}_{B_{f},z}^{a_{f},x_{f}})\otimes\lbrack y].
\end{align}

\begin{notation}
In the above and in what follows, we employ a shorthand $[x]\equiv
|x\rangle\langle x|_{X}$ or $[a]\equiv|a\rangle\langle a|_{\overline{A}}$, etc.
\end{notation}

\noindent The above states can be embedded in the following ones:%
\begin{align}
\omega_{X_{f}\overline{A}_{f}B_{f}^{\prime}YZ} &  :=\sum_{z}\omega_{X_{f}\overline
{A}_{f}B_{f}^{\prime}Y}^{z}\otimes p(z)[z],\\
\tau_{X_{f}\overline{A}_{f}B_{f}^{\prime}YZ} &  :=\sum_{z}\tau_{X_{f}\overline{A}%
_{f}B_{f}^{\prime}Y}^{z}\otimes q(z)[z].
\end{align}
The states above are extended by the following ones:%
\begin{align}
\omega_{X_{f}X\overline{A}_{f}\overline{A}B_{f}^{\prime}Y}^{z} &  :=\sum_{a_{f}%
,a,x,x_{f},y}p(x_{f}|y,z)[x_{f}]\otimes p(x|x_{f},z)[x]\otimes p(a_{f}%
|x_{f},x,a,z)[a_{f}]\nonumber\\
&  \qquad\qquad\otimes\lbrack a]\otimes\frac{\mathcal{L}_{y}^{(z)}%
(\mathcal{K}_{z}(\hat{\rho}_{B}^{a,x}))}{p(z)}\otimes\lbrack y],\\
\tau_{X_{f}X\overline{A}_{f}\overline{A}B_{f}^{\prime}Y}^{z} &  :=\sum_{a_{f}%
,a,x,x_{f},y}p(x_{f}|y,z)[x_{f}]\otimes p(x|x_{f},z)[x]\otimes p(a_{f}%
|x_{f},x,a,z)[a_{f}]\nonumber\\
&  \qquad\qquad\otimes\lbrack a]\otimes\frac{\mathcal{L}_{y}^{(z)}%
(\mathcal{K}_{z}(\hat{\sigma}_{B}^{a,x}))}{q(z)}\otimes\lbrack y],
\end{align}
which in turn are elements of the following classical--quantum states:%
\begin{align}
\omega_{X_{f}X\overline{A}_{f}\overline{A}B_{f}^{\prime}YZ} &  :=\sum_{z}\omega
_{X_{f}X\overline{A}_{f}\overline{A}B_{f}^{\prime}Y}^{z}\otimes p(z)[z],\\
\tau_{X_{f}X\overline{A}_{f}\overline{A}B_{f}^{\prime}YZ} &  :=\sum_{z}\tau_{X_{f}%
X\overline{A}_{f}\overline{A}B_{f}^{\prime}Y}^{z}\otimes q(z)[z].
\end{align}
Consider that%
\begin{align}
&  \!\!\!\!\!\!\sum_{z}p(z)\inf_{\hat{\zeta}^{z}\in\operatorname{LHS}}%
D(\omega_{X_{f}\overline{A}_{f}B_{f}^{\prime}Y}^{z}\Vert\zeta_{X_{f}\overline{A}%
_{f}B_{f}^{\prime}Y}^{z})\nonumber\\
&  \leq\sum_{z}p(z)D(\omega_{X_{f}\overline{A}_{f}B_{f}^{\prime}Y}^{z}\Vert
\tau_{X_{f}\overline{A}_{f}B_{f}^{\prime}Y}^{z})\\
&  \leq\sum_{z}p(z)D(\omega_{X_{f}\overline{A}_{f}B_{f}^{\prime}Y}^{z}\Vert
\tau_{X_{f}\overline{A}_{f}B_{f}^{\prime}Y}^{z})+D(p\Vert q)\\
&  =D(\omega_{X_{f}\overline{A}_{f}B_{f}^{\prime}YZ}\Vert\tau_{X_{f}\overline{A}%
_{f}B_{f}^{\prime}YZ})\\
&  \leq D(\omega_{X_{f}X\overline{A}_{f}\overline{A}B_{f}^{\prime}YZ}\Vert\tau
_{X_{f}X\overline{A}_{f}\overline{A}B_{f}^{\prime}YZ})\\
&  =D(\omega_{X_{f}X\overline{A}B_{f}^{\prime}YZ}\Vert\tau_{X_{f}X\overline{A}%
B_{f}^{\prime}YZ}).
\end{align}
In the first line, we take $\hat{\zeta}^{z}$ to denote a general
LHS\ assemblage $\{\hat{\zeta}_{B_{f}}^{a_{f},x_{f},z}\}_{a_{f},x_{f}}$ and
$\zeta_{X_{f}\overline{A}_{f}B_{f}^{\prime}Y}^{z}$ denotes the following state:%
\begin{equation}
\zeta_{X_{f}\overline{A}_{f}B_{f}^{\prime}Y}^{z}:=\sum_{a_{f},x_{f},y}%
p(x_{f}|y,z)[x_{f}]\otimes\lbrack a_{f}]\otimes\mathcal{L}_{y}^{(z)}(\hat
{\zeta}_{B_{f}}^{a_{f},x_{f},z})\otimes\lbrack y].
\end{equation}
The first inequality follows by considering that the state $\tau_{X_{f}\overline
{A}_{f}B_{f}^{\prime}Y}^{z}$ arises from the action of the $z$-dependent
1W-LOCC\ operation $\{p_{X_{f}|YZ=z},\{\mathcal{L}_{y}^{(z)}\}_{y}\}$\ on the
LHS\ assemblage $\{\hat{\sigma}_{B_{f},z}^{a_{f},x_{f}}\}_{a_{f},x_{f}}$. The
second inequality follows from non-negativity of relative entropy. The first
equality is a consequence of the property of relative entropy recalled in
\eqref{eq:rel-ent-block-prop}. The final inequality follows from the data
processing inequality for quantum relative entropy, and the final equality
follows because the random variable in $\overline{A}_{f}$ for each state
$\omega_{X_{f}X\overline{A}_{f}\overline{A}B_{f}^{\prime}YZ}$ and $\tau_{X_{f}X\overline
{A}_{f}\overline{A}B_{f}^{\prime}YZ}$\ is produced by the same classical channel
$p(a_{f}|x_{f},x,a,z)$, so that we get the inequality $D(\omega_{X_{f}X\overline
{A}_{f}\overline{A}B_{f}^{\prime}YZ}\Vert\tau_{X_{f}X\overline{A}_{f}\overline{A}%
B_{f}^{\prime}YZ})\leq D(\omega_{X_{f}X\overline{A}B_{f}^{\prime}YZ}\Vert
\tau_{X_{f}X\overline{A}B_{f}^{\prime}YZ})$ by data processing and the opposite
inequality follows by taking a partial trace over system $\overline{A}_{f}$.

We have shown that the above chain of inequalities holds for all assemblages
$\{\hat{\sigma}_{B}^{a,x}\}_{a,x}\in\operatorname{LHS}$, and so we can
conclude that%
\begin{equation}
\sum_{z}p(z)\inf_{\zeta^{z}\in\operatorname{LHS}}D(\omega_{X_{f}\overline{A}%
_{f}B_{f}^{\prime}Y}^{z}\Vert\zeta_{X_{f}\overline{A}_{f}B_{f}^{\prime}Y}^{z}%
)\leq\inf_{\{\hat{\sigma}_{B}^{a,x}\}_{a,x}\in\operatorname{LHS}}%
D(\omega_{X_{f}X\overline{A}B_{f}^{\prime}YZ}\Vert\tau_{X_{f}X\overline{A}B_{f}^{\prime
}YZ}).
\end{equation}
The above inequality holds for all 1W-LOCC\ strategies $\{p_{X_{f}%
|YZ=z},\{\mathcal{L}_{y}^{(z)}\}_{y}\}_{z}$, so we can now take a supremum
over all such\ strategies $\{p_{X_{f}|YZ=z},\{\mathcal{L}_{y}^{(z)}%
\}_{y}\}_{z}$ to find that%
\begin{align}
&  \!\!\!\!\!\!\sup_{\{p_{X_{f}|YZ=z},\{\mathcal{L}_{y}^{(z)}\}_{y}\}_{z}}%
\sum_{z}p(z)\inf_{\zeta^{z}\in\operatorname{LHS}}D(\omega_{X_{f}\overline{A}%
_{f}B_{f}^{\prime}Y}^{z}\Vert\zeta_{X_{f}\overline{A}_{f}B_{f}^{\prime}Y}%
^{z})\nonumber\\
&  \leq\sup_{\{p_{X_{f}|YZ=z},\{\mathcal{L}_{y}^{(z)}\}_{y}\}_{z}}\inf
_{\{\hat{\sigma}_{B}^{a,x}\}_{a,x}\in\operatorname{LHS}}D(\omega_{X_{f}%
X\overline{A}B_{f}^{\prime}YZ}\Vert\tau_{X_{f}X\overline{A}B_{f}^{\prime}YZ})\\
&  \leq R_{S}(\overline{A};B)_{\hat{\rho}}.
\end{align}
The last inequality follows because $\{p_{X_{f}|YZ=z},\{\mathcal{L}_{y}%
^{(z)}\}_{y}\}_{z}$ is a particular 1W-LOCC\ strategy, while $R_{S}(\overline
{A};B)_{\hat{\rho}}$ involves an optimization over all 1W-LOCC\ strategies.
The quantity on the first line above can be rewritten as%
\begin{align}
&  \!\!\!\!\!\!\sup_{\{p_{X_{f}|YZ=z},\{\mathcal{L}_{y}^{(z)}\}_{y}\}_{z}}%
\sum_{z}p(z)\inf_{\zeta^{z}\in\operatorname{LHS}}D(\omega_{X_{f}\overline{A}%
_{f}B_{f}^{\prime}Y}^{z}\Vert\zeta_{X_{f}\overline{A}_{f}B_{f}^{\prime}Y}%
^{z})\nonumber\\
&  =\sum_{z}p(z)\sup_{\{p_{X_{f}|YZ=z},\{\mathcal{L}_{y}^{(z)}\}_{y}\}}%
\inf_{\zeta^{z}\in\operatorname{LHS}}D(\omega_{X_{f}\overline{A}_{f}B_{f}^{\prime
}Y}^{z}\Vert\zeta_{X_{f}\overline{A}_{f}B_{f}^{\prime}Y}^{z})\\
&  =\sum_{z}p(z)R_{S}(\overline{A}_{f};B_{f})_{\hat{\rho}_{z}}.
\end{align}
This concludes the proof.
\end{proof}

\subsection{Convexity\label{sec:convexity}}

Here we prove that the relative entropy of steering is convex with respect to
the assemblages on which it is evaluated.

\begin{proposition}
[Convexity]Let $\lambda\in\lbrack0,1]$. Let $\{\hat{\rho}_{B}^{a,x}\}_{a,x}$
and $\{\hat{\theta}_{B}^{a,x}\}_{a,x}$ be two assemblages, and consider an
assemblage $\{\hat{\tau}_{B}^{a,x}:=\lambda\hat{\rho}_{B}^{a,x}+(1-\lambda
)\hat{\theta}_{B}^{a,x}\}_{a,x}$. The restricted relative entropy of steering
is convex in the following sense:%
\begin{equation}
R_{S}(\overline{A};B)_{\hat{\tau}}\leq\lambda R_{S}(\overline{A};B)_{\hat{\rho}%
}+(1-\lambda)R_{S}(\overline{A};B)_{\hat{\theta}}.\label{convexity}%
\end{equation}

\end{proposition}

\begin{proof}
Let $\{p_{X|Y},\{\mathcal{K}_{y}\}_{y}\}$ denote an arbitrary
1W-LOCC\ operation, let $\{\hat{\sigma}_{B}^{a,x}\}_{a,x}$ and $\{\hat{\omega
}_{B}^{a,x}\}_{a,x}$ be arbitrary LHS\ assemblages. Consider the following
states:%
\begin{align}
\rho_{X\overline{A}B^{\prime}Y} &  :=\sum_{x,a,y}p_{X|Y}(x|y)[x]\otimes\lbrack
a]\otimes\mathcal{K}_{y}(\hat{\rho}_{B}^{a,x})\otimes\lbrack y],\\
\theta_{X\overline{A}B^{\prime}Y} &  :=\sum_{x,a,y}p_{X|Y}(x|y)[x]\otimes\lbrack
a]\otimes\mathcal{K}_{y}(\hat{\theta}_{B}^{a,x})\otimes\lbrack y],\\
\sigma_{X\overline{A}B^{\prime}Y} &  :=\sum_{x,a,y}p_{X|Y}(x|y)[x]\otimes\lbrack
a]\otimes\mathcal{K}_{y}(\hat{\sigma}_{B}^{a,x})\otimes\lbrack y],\\
\omega_{X\overline{A}B^{\prime}Y} &  :=\sum_{x,a,y}p_{X|Y}(x|y)[x]\otimes\lbrack
a]\otimes\mathcal{K}_{y}(\hat{\omega}_{B}^{a,x})\otimes\lbrack y].
\end{align}
Let us define the following states:%
\begin{align}
\zeta_{QX\overline{A}B^{\prime}Y} &  :=\lambda|0\rangle\langle0|_{Q}\otimes
\rho_{X\overline{A}B^{\prime}Y}+(1-\lambda)|1\rangle\langle1|_{Q}\otimes
\theta_{X\overline{A}B^{\prime}Y},\\
\kappa_{QX\overline{A}B^{\prime}Y} &  :=\lambda|0\rangle\langle0|_{Q}\otimes
\sigma_{X\overline{A}B^{\prime}Y}+(1-\lambda)|1\rangle\langle1|_{Q}\otimes
\omega_{X\overline{A}B^{\prime}Y}.
\end{align}
Consider that%
\begin{equation}
\zeta_{X\overline{A}B^{\prime}Y}=\operatorname{Tr}_{Q}(\zeta_{QX\overline{A}B^{\prime}%
Y})=\sum_{x,a,y}p_{X|Y}(x|y)[x]\otimes\lbrack a]\otimes\mathcal{K}_{y}%
(\hat{\tau}_{B}^{a,x})\otimes\lbrack y].
\end{equation}
Then we have the following chain of inequalities:%
\begin{align}
&  \!\!\!\!\!\!\lambda D(\rho_{X\overline{A}B^{\prime}Y}\Vert\sigma_{X\overline
{A}B^{\prime}Y})+(1-\lambda)D(\theta_{X\overline{A}B^{\prime}Y}\Vert\omega
_{X\overline{A}B^{\prime}Y})\nonumber\\
&  =D(\zeta_{QX\overline{A}B^{\prime}Y}\Vert\kappa_{QX\overline{A}B^{\prime}Y})\\
&  \geq D(\zeta_{X\overline{A}B^{\prime}Y}\Vert\kappa_{X\overline{A}B^{\prime}Y})\\
&  \geq\inf_{\hat{\varsigma}\in\operatorname{LHS}}D(\zeta_{X\overline{A}B^{\prime
}Y}\Vert\varsigma_{X\overline{A}B^{\prime}Y}).
\end{align}
In the first equality, we have exploited the property of quantum relative
entropy in \eqref{eq:rel-ent-block-prop}. The first inequality follows from
the data processing inequality for quantum relative entropy, by tracing over
system $Q$. The final inequality follows by defining the LHS\ assemblage
$\hat{\varsigma}\equiv\{\hat{\varsigma}_{B}^{a,x}\}_{a,x}$, the corresponding
state%
\begin{equation}
\varsigma_{X\overline{A}B^{\prime}Y}:=\sum_{x,a,y}p_{X|Y}(x|y)[x]\otimes\lbrack
a]\otimes\mathcal{K}_{y}(\hat{\varsigma}_{B}^{a,x})\otimes\lbrack y],
\end{equation}
and taking an infimum with respect to all such LHS\ assemblages. Since we have
shown that the above inequality holds for all LHS\ assemblages $\{\hat{\sigma
}_{B}^{a,x}\}_{a,x}$ and $\{\hat{\omega}_{B}^{a,x}\}_{a,x}$, we can conclude
that%
\begin{equation}
\inf_{\hat{\varsigma}\in\operatorname{LHS}}D(\zeta_{X\overline{A}B^{\prime}Y}%
\Vert\varsigma_{X\overline{A}B^{\prime}Y})\leq\lambda\inf_{\hat{\sigma}%
\in\operatorname{LHS}}D(\rho_{X\overline{A}B^{\prime}Y}\Vert\sigma_{X\overline
{A}B^{\prime}Y})+(1-\lambda)\inf_{\hat{\omega}\in\operatorname{LHS}}%
D(\theta_{X\overline{A}B^{\prime}Y}\Vert\omega_{X\overline{A}B^{\prime}Y}).
\end{equation}
Finally, since we have shown that the above inequality holds for an arbitrary
1W-LOCC\ operation $\{p_{X|Y},\{\mathcal{K}_{y}\}_{y}\}$, we can conclude that%
\begin{align}
&  \sup_{\{p_{X|Y},\{\mathcal{K}_{y}\}_{y}\}}\inf_{\hat{\varsigma}%
\in\operatorname{LHS}}D(\zeta_{X\overline{A}B^{\prime}Y}\Vert\varsigma_{X\overline
{A}B^{\prime}Y})\nonumber\\
&  \leq\sup_{\{p_{X|Y},\{\mathcal{K}_{y}\}_{y}\}}\left[  \lambda\inf
_{\hat{\sigma}\in\operatorname{LHS}}D(\rho_{X\overline{A}B^{\prime}Y}\Vert
\sigma_{X\overline{A}B^{\prime}Y})+(1-\lambda)\inf_{\hat{\omega}\in
\operatorname{LHS}}D(\theta_{X\overline{A}B^{\prime}Y}\Vert\omega_{X\overline
{A}B^{\prime}Y})\right]  \\
&  \leq\lambda\sup_{\{p_{X|Y},\{\mathcal{K}_{y}\}_{y}\}}\inf_{\hat{\sigma}%
\in\operatorname{LHS}}D(\rho_{X\overline{A}B^{\prime}Y}\Vert\sigma_{X\overline
{A}B^{\prime}Y})+(1-\lambda)\sup_{\{p_{X|Y},\{\mathcal{K}_{y}\}_{y}\}}%
\inf_{\hat{\omega}\in\operatorname{LHS}}D(\theta_{X\overline{A}B^{\prime}Y}%
\Vert\omega_{X\overline{A}B^{\prime}Y}).
\end{align}
This final inequality is equivalent to the one in the statement of the proposition.
\end{proof}

\subsection{Upper bounds on relative entropy of
steering\label{sec:upper-bounds}}

\begin{proposition}
[Upper bounds]\label{prop:rel-ent-st-up-bnd}Let $\{\hat{\rho}_{B}%
^{a,x}\}_{a,x}$ be an assemblage. Then%
\begin{equation}
R_{S}(\overline{A};B)_{\hat{\rho}}\leq\sup_{\{p_{X|Y},\{\mathcal{K}_{y}\}_{y}%
\}}I(XB^{\prime}Y;\overline{A})_{\rho}\leq\sup_{p_{X}}H(\overline{A})\leq\log_{2}%
|\overline{A}|,
\end{equation}
where the mutual information is with respect to the following state:%
\begin{equation}
\rho_{X\overline{A}B^{\prime}Y}:=\sum_{x,a,y}p_{X|Y}(x|y)[x]\otimes\lbrack
a]\otimes\mathcal{K}_{y}(\hat{\rho}_{B}^{a,x})\otimes\lbrack y],
\end{equation}
and the entropy $H(\overline{A})$\ is with respect to the probability distribution
$p_{\overline{A}}(a):=\sum_{x}p_{X}(x)\operatorname{Tr}(\hat{\rho}_{B}^{a,x})$.
\end{proposition}

\begin{proof}
From the definition of relative entropy of steering, we have that%
\begin{equation}
R_{S}(\overline{A};B)_{\hat{\rho}}=\sup_{\{p_{X|Y},\{\mathcal{K}_{y}\}_{y}\}}%
\inf_{\{\hat{\sigma}_{B}^{a,x}\}_{a,x}\in\operatorname{LHS}}D(\rho_{X\overline
{A}B^{\prime}Y}\Vert\sigma_{X\overline{A}B^{\prime}Y}),
\end{equation}
where%
\begin{equation}
\sigma_{X\overline{A}B^{\prime}Y}:=\sum_{x,a,y}p_{X|Y}(x|y)[x]\otimes\lbrack
a]\otimes\mathcal{K}_{y}(\hat{\sigma}_{B}^{a,x})\otimes\lbrack y].
\end{equation}
Consider the probability distribution on system $\overline{A}$ that results from
partial trace with respect to the state $\rho_{X\overline{A}B^{\prime}Y}$:%
\begin{align}
\rho_{\overline{A}} &  =\operatorname{Tr}_{XB^{\prime}Y}(\rho_{X\overline{A}B^{\prime}%
Y})\\
&  =\sum_{a}\left[  \sum_{x,y}p_{X|Y}(x|y)\operatorname{Tr}(\mathcal{K}%
_{y}(\hat{\rho}_{B}^{a,x}))\right]  [a].
\end{align}
Then define $p_{A}(a):=\sum_{x,y}p_{X|Y}(x|y)\operatorname{Tr}(\mathcal{K}%
_{y}(\hat{\rho}_{B}^{a,x}))$. Also, take $\rho_{B}:=\sum_{a}\hat{\rho}%
_{B}^{a,x}$. A particular assemblage with a local-hidden-state model is the
following one:%
\begin{equation}
\{\xi_{B}^{a,x}:=p_{A}(a)\rho_{B}\}_{a,x}.
\end{equation}
This particular LHS\ assemblage leads to the following state on systems
$X\overline{A}B^{\prime}Y$:%
\begin{equation}
\sum_{x,a,y}p_{X|Y}(x|y)[x]\otimes p_{A}(a)[a]\otimes\mathcal{K}_{y}(\rho
_{B})\otimes\lbrack y]=\rho_{\overline{A}}\otimes\rho_{XB^{\prime}Y},
\end{equation}
where the states on the right are the marginals of $\rho_{X\overline{A}B^{\prime}%
Y}$. Then%
\begin{align}
&  \!\!\!\!\!\!\sup_{\{p_{X|Y},\{\mathcal{K}_{y}\}_{y}\}}\inf_{\{\hat{\sigma
}_{B}^{a,x}\}_{a,x}\in\operatorname{LHS}}D(\rho_{X\overline{A}B^{\prime}Y}%
\Vert\sigma_{X\overline{A}B^{\prime}Y})\nonumber\\
&  \leq\sup_{\{p_{X|Y},\{\mathcal{K}_{y}\}_{y}\}}D(\rho_{X\overline{A}B^{\prime}%
Y}\Vert\rho_{\overline{A}}\otimes\rho_{XB^{\prime}Y})\\
&  =\sup_{\{p_{X|Y},\{\mathcal{K}_{y}\}_{y}\}}I(XB^{\prime}Y;\overline{A})_{\rho}\\
&  \leq\sup_{p_{X}}H(\overline{A})_{\rho}\\
&  \leq\log_{2}|\overline{A}|.
\end{align}
The first inequality follows because we can choose a particular
LHS\ assemblage and get an upper bound on $R_{S}(\overline{A};B)_{\hat{\rho}}$. The
first equality follows from the well known characterization of quantum mutual
information as the quantum relative entropy between the joint state and the
product of the marginals. The second inequality follows because $I(XB^{\prime
}Y;\overline{A})_{\rho}\leq H(\overline{A})_{\rho}$, given that system $\overline{A}$ is
classical, and then we can optimize this quantity with respect to all possible
input distributions $p_{X}$. The final inequality is a well known dimension
bound for entropy.
\end{proof}

\subsection{Continuity\label{sec:cont}}

Before giving the statement of continuity, let us first define the
(normalized)\ trace distance of assemblages as follows:

\begin{definition}
[Trace distance of assemblages]Let $\{\hat{\rho}_{B}^{a,x}\}_{a,x}$ and
$\{\hat{\theta}_{B}^{a,x}\}_{a,x}$ be two assemblages. We define the
normalized trace distance of assemblages as%
\begin{equation}
\Delta(\hat{\rho},\hat{\theta}):=\frac{1}{2}\sup_{\{p_{X|Y},\{\mathcal{K}%
_{y}\}_{y}\}}\left\Vert \rho_{X\overline{A}B^{\prime}Y}-\theta_{X\overline{A}B^{\prime
}Y}\right\Vert _{1},
\end{equation}
where $\Vert C\Vert_{1}:=\operatorname{Tr}(\sqrt{C^{\dag}C})$ and
\begin{align}
\rho_{X\overline{A}B^{\prime}Y} &  :=\sum_{x,a,y}p_{X|Y}(x|y)[x]\otimes\lbrack
a]\otimes\mathcal{K}_{y}(\hat{\rho}_{B}^{a,x})\otimes\lbrack y],\\
\theta_{X\overline{A}B^{\prime}Y} &  :=\sum_{x,a,y}p_{X|Y}(x|y)[x]\otimes\lbrack
a]\otimes\mathcal{K}_{y}(\hat{\theta}_{B}^{a,x})\otimes\lbrack y].
\end{align}

\end{definition}

By properties of trace distance, it follows that $\Delta(\hat{\rho
},\hat{\theta})\in\lbrack0,1]$. Furthermore, given that the trace distance of assemblages represents a measure of distinguishability of two different assemblages, the above definition involves an optimization over all 1W-LOCC strategies that could be used to distinguish them.

\begin{proposition}
[Metric]The trace distance of assemblages is a metric, in the sense that for
any three assemblages $\{\hat{\rho}_{B}^{a,x}\}_{a,x}$, $\{\hat{\theta}%
_{B}^{a,x}\}_{a,x}$, and $\{\hat{\omega}_{B}^{a,x}\}_{a,x}$:%
\begin{align}
\Delta(\hat{\rho},\hat{\theta})  &  \geq0,\label{eq:metric-1}\\
\Delta(\hat{\rho},\hat{\theta})  &  =0\text{ if and only if }\hat{\rho}%
_{B}^{a,x}=\hat{\theta}_{B}^{a,x}\text{ for all }a,x,\label{eq:metric-2}\\
\Delta(\hat{\rho},\hat{\theta})  &  =\Delta(\hat{\theta},\hat{\rho
}),\label{eq:metric-3}\\
\Delta(\hat{\rho},\hat{\theta})  &  \leq\Delta(\hat{\rho},\hat{\omega}%
)+\Delta(\hat{\omega},\hat{\theta}). \label{eq:metric-4}%
\end{align}

\end{proposition}

\begin{proof}
These properties follow directly from the fact that normalized trace distance
is a metric for quantum states. We give brief proofs for completeness. The
inequality in \eqref{eq:metric-1} follows because the normalized trace
distance is non-negative. Regarding \eqref{eq:metric-2}, the implication
$\hat{\rho}_{B}^{a,x}=\hat{\theta}_{B}^{a,x}$ for all $a,x$ $\Longrightarrow$
$\Delta(\hat{\rho},\hat{\theta})=0$ follows because the states resulting from
an arbitrary 1W-LOCC\ operation are the same if the assemblages are the same.
To see the other implication, consider that $\Delta(\hat{\rho},\hat{\theta
})=0$ means that the normalized trace distance between $\rho_{X\overline
{A}B^{\prime}Y}$ and $\theta_{X\overline{A}B^{\prime}Y}$ is equal to zero for all
possible 1W-LOCC\ operations. So we can pick the 1W-LOCC\ operation to be a
uniform distribution over the input $x$ and the identity channel on system $B$
and find that%
\begin{align}
0 &  =\left\Vert \sum_{x,a}\frac{1}{\left\vert \mathcal{X}\right\vert
}[x]\otimes\lbrack a]\otimes\hat{\rho}_{B}^{a,x}-\sum_{x,a}\frac{1}{\left\vert
\mathcal{X}\right\vert }[x]\otimes\lbrack a]\otimes\hat{\theta}_{B}%
^{a,x}\right\Vert _{1}\\
&  =\sum_{x,a}\frac{1}{\left\vert \mathcal{X}\right\vert }\left\Vert \hat
{\rho}_{B}^{a,x}-\hat{\theta}_{B}^{a,x}\right\Vert _{1}.
\end{align}
By the fact that the normalized trace distance is a metric, we can then
conclude that $\hat{\rho}_{B}^{a,x}=\hat{\theta}_{B}^{a,x}$ for all $a,x$. The
equality in \eqref{eq:metric-3} clearly holds. The triangle inequality in
\eqref{eq:metric-4} follows because normalized trace distance obeys the
triangle inequality:%
\begin{align}
\Delta(\hat{\rho},\hat{\theta}) &  =\frac{1}{2}\sup_{\{p_{X|Y},\{\mathcal{K}%
_{y}\}_{y}\}}\left\Vert \rho_{X\overline{A}B^{\prime}Y}-\theta_{X\overline{A}B^{\prime
}Y}\right\Vert _{1}\\
&  \leq\frac{1}{2}\sup_{\{p_{X|Y},\{\mathcal{K}_{y}\}_{y}\}}\left[  \left\Vert
\rho_{X\overline{A}B^{\prime}Y}-\omega_{X\overline{A}B^{\prime}Y}\right\Vert
_{1}+\left\Vert \omega_{X\overline{A}B^{\prime}Y}-\theta_{X\overline{A}B^{\prime}%
Y}\right\Vert _{1}\right]  \\
&  \leq\frac{1}{2}\sup_{\{p_{X|Y},\{\mathcal{K}_{y}\}_{y}\}}\left\Vert
\rho_{X\overline{A}B^{\prime}Y}-\omega_{X\overline{A}B^{\prime}Y}\right\Vert _{1}%
+\frac{1}{2}\sup_{\{p_{X|Y},\{\mathcal{K}_{y}\}_{y}\}}\left\Vert \omega
_{X\overline{A}B^{\prime}Y}-\theta_{X\overline{A}B^{\prime}Y}\right\Vert _{1}\\
&  =\Delta(\hat{\rho},\hat{\omega})+\Delta(\hat{\omega},\hat{\theta}).
\end{align}
This concludes the proof.
\end{proof}

\begin{theorem}
[Uniform continuity]\label{thm:continuity}Let $\{\hat{\rho}_{B}^{a,x}\}_{a,x}$
and $\{\hat{\theta}_{B}^{a,x}\}_{a,x}$ be assemblages such that
$\Delta(\hat{\rho},\hat{\theta})\leq\varepsilon\in\left[  0,1\right]  $. Then%
\begin{equation}
\left\vert R_{S}(\overline{A};B)_{\hat{\rho}}-R_{S}(\overline{A};B)_{\hat{\theta}%
}\right\vert \leq\varepsilon\log_{2}|\overline{A}|+g(\varepsilon),
\end{equation}
where $g(\varepsilon):=(\varepsilon+1)\log_{2}(\varepsilon+1)-\varepsilon
\log_{2}\varepsilon$.
\end{theorem}

\begin{proof}
We note that the following proof is very similar to those of \cite[Lemmas~2
and 7]{Winter2015}, but we give a detailed proof for completeness. Let
$\{p_{X|Y},\{\mathcal{K}_{y}\}_{y}\}$ be an arbitrary 1W-LOCC\ operation, and
let $\{\hat{\sigma}_{B}^{a,x}\}_{a,x}$ and $\{\hat{\omega}_{B}^{a,x}\}_{a,x}$
be arbitrary LHS\ assemblages. Consider the following states:%
\begin{align}
\rho_{X\overline{A}B^{\prime}Y} &  :=\sum_{x,a,y}p_{X|Y}(x|y)[x]\otimes\lbrack
a]\otimes\mathcal{K}_{y}(\hat{\rho}_{B}^{a,x})\otimes\lbrack y],\\
\theta_{X\overline{A}B^{\prime}Y} &  :=\sum_{x,a,y}p_{X|Y}(x|y)[x]\otimes\lbrack
a]\otimes\mathcal{K}_{y}(\hat{\theta}_{B}^{a,x})\otimes\lbrack y],\\
\sigma_{X\overline{A}B^{\prime}Y} &  :=\sum_{x,a,y}p_{X|Y}(x|y)[x]\otimes\lbrack
a]\otimes\mathcal{K}_{y}(\hat{\sigma}_{B}^{a,x})\otimes\lbrack y],\\
\omega_{X\overline{A}B^{\prime}Y} &  :=\sum_{x,a,y}p_{X|Y}(x|y)[x]\otimes\lbrack
a]\otimes\mathcal{K}_{y}(\hat{\omega}_{B}^{a,x})\otimes\lbrack y].
\end{align}
Consider that $\frac{1}{2}\left\Vert \rho_{X\overline{A}B^{\prime}Y}-\theta
_{X\overline{A}B^{\prime}Y}\right\Vert _{1}\leq\varepsilon$ by assumption. Let us
set $\varepsilon_{0}:=\frac{1}{2}\left\Vert \rho_{X\overline{A}B^{\prime}Y}%
-\theta_{X\overline{A}B^{\prime}Y}\right\Vert _{1}$. If $\varepsilon_{0}=0$, then
the particular 1W-LOCC\ operation cannot distinguish the states from each
other, so that $\rho_{X\overline{A}B^{\prime}Y}=\theta_{X\overline{A}B^{\prime}Y}$, and
we find that%
\begin{equation}
\inf_{\hat{\sigma}\in\operatorname{LHS}}D(\rho_{X\overline{A}B^{\prime}Y}%
\Vert\sigma_{X\overline{A}B^{\prime}Y})=\inf_{\hat{\sigma}\in\operatorname{LHS}%
}D(\theta_{X\overline{A}B^{\prime}Y}\Vert\sigma_{X\overline{A}B^{\prime}Y})
\end{equation}
in this case, so that there is nothing to prove. So let us instead suppose
that $\varepsilon_{0}\neq0$ and define%
\begin{equation}
\Delta_{X\overline{A}B^{\prime}Y}:=\frac{1}{\varepsilon_{0}}\left(  \rho_{X\overline
{A}B^{\prime}Y}-\theta_{X\overline{A}B^{\prime}Y}\right)  _{+},
\end{equation}
where $(\cdot)_{+}$ indicates the positive part of $\rho_{X\overline{A}B^{\prime}%
Y}-\theta_{X\overline{A}B^{\prime}Y}$. Since $\rho_{X\overline{A}B^{\prime}Y}%
-\theta_{X\overline{A}B^{\prime}Y}$ is traceless and its trace norm is equal to
$2\varepsilon_{0}$, it follows that $\Delta_{X\overline{A}B^{\prime}Y}$ is a
density operator. Consider that%
\begin{align}
\rho_{X\overline{A}B^{\prime}Y} &  =\theta_{X\overline{A}B^{\prime}Y}+\left(
\rho_{X\overline{A}B^{\prime}Y}-\theta_{X\overline{A}B^{\prime}Y}\right)  \\
&  \leq\theta_{X\overline{A}B^{\prime}Y}+\varepsilon_{0}\Delta_{X\overline{A}B^{\prime
}Y}\\
&  =\left(  1+\varepsilon_{0}\right)  \left(  \frac{1}{1+\varepsilon_{0}%
}\theta_{X\overline{A}B^{\prime}Y}+\frac{\varepsilon_{0}}{1+\varepsilon_{0}}%
\Delta_{X\overline{A}B^{\prime}Y}\right)  \\
&  =:\left(  1+\varepsilon_{0}\right)  \zeta_{X\overline{A}B^{\prime}Y}.
\end{align}
Setting%
\begin{equation}
\Delta_{X\overline{A}B^{\prime}Y}^{\prime}:=\frac{1}{\varepsilon_{0}}\left[
\left(  1+\varepsilon_{0}\right)  \zeta_{X\overline{A}B^{\prime}Y}-\rho_{X\overline
{A}B^{\prime}Y}\right]  ,
\end{equation}
we see that $\Delta_{X\overline{A}B^{\prime}Y}^{\prime}$ is a density operator
also, satisfying%
\begin{align}
\zeta_{X\overline{A}B^{\prime}Y} &  =\frac{1}{1+\varepsilon_{0}}\theta_{X\overline
{A}B^{\prime}Y}+\frac{\varepsilon_{0}}{1+\varepsilon_{0}}\Delta_{X\overline
{A}B^{\prime}Y}\\
&  =\frac{1}{1+\varepsilon_{0}}\rho_{X\overline{A}B^{\prime}Y}+\frac{\varepsilon
_{0}}{1+\varepsilon_{0}}\Delta_{X\overline{A}B^{\prime}Y}^{\prime}\ .
\end{align}
From the joint convexity of relative entropy, we find that%
\begin{align}
&  \!\!\!\!\!\!\inf_{\hat{\kappa}\in\operatorname{LHS}}D(\zeta_{X\overline
{A}B^{\prime}Y}\Vert\kappa_{X\overline{A}B^{\prime}Y})\nonumber\\
&  \leq D(\zeta_{X\overline{A}B^{\prime}Y}\Vert\left[  1+\varepsilon_{0}\right]
^{-1}\sigma_{X\overline{A}B^{\prime}Y}+\varepsilon_{0}\left[  1+\varepsilon
_{0}\right]  ^{-1}\omega_{X\overline{A}B^{\prime}Y})\\
&  \leq\frac{1}{1+\varepsilon_{0}}D(\theta_{X\overline{A}B^{\prime}Y}\Vert
\sigma_{X\overline{A}B^{\prime}Y})+\frac{\varepsilon_{0}}{1+\varepsilon_{0}%
}D(\Delta_{X\overline{A}B^{\prime}Y}\Vert\omega_{X\overline{A}B^{\prime}Y}).
\end{align}
Since $\sigma_{X\overline{A}B^{\prime}Y}$ and $\omega_{X\overline{A}B^{\prime}Y}$ are
states arising from arbitrary LHS\ assemblages, we can conclude that%
\begin{multline}
\inf_{\hat{\kappa}\in\operatorname{LHS}}D(\zeta_{X\overline{A}B^{\prime}Y}%
\Vert\kappa_{X\overline{A}B^{\prime}Y})\leq\label{eq:cont-up-bound}\\
\frac{1}{1+\varepsilon_{0}}\inf_{\hat{\sigma}\in\operatorname{LHS}}%
D(\theta_{X\overline{A}B^{\prime}Y}\Vert\sigma_{X\overline{A}B^{\prime}Y})+\frac
{\varepsilon_{0}}{1+\varepsilon_{0}}\inf_{\hat{\omega}\in\operatorname{LHS}%
}D(\Delta_{X\overline{A}B^{\prime}Y}\Vert\omega_{X\overline{A}B^{\prime}Y}).
\end{multline}
Now consider that for a state $\kappa_{X\overline{A}B^{\prime}Y}$ arising from an
arbitrary LHS\ assemblage $\hat{\kappa}$, we have that%
\begin{align}
&  \!\!\!\!\!\!D(\zeta_{X\overline{A}B^{\prime}Y}\Vert\kappa_{X\overline{A}B^{\prime}%
Y})\nonumber\\
&  =-H(\zeta_{X\overline{A}B^{\prime}Y})-\operatorname{Tr}(\zeta_{X\overline{A}%
B^{\prime}Y}\log_{2}\kappa_{X\overline{A}B^{\prime}Y})\\
&  \geq-h_{2}(\varepsilon_{0}/\left[  1+\varepsilon_{0}\right]  )-\frac
{1}{1+\varepsilon_{0}}H(\rho_{X\overline{A}B^{\prime}Y})-\frac{\varepsilon_{0}%
}{1+\varepsilon_{0}}H(\Delta_{X\overline{A}B^{\prime}Y}^{\prime})\nonumber\\
&  \qquad-\frac{1}{1+\varepsilon_{0}}\operatorname{Tr}(\rho_{X\overline{A}%
B^{\prime}Y}\log\kappa_{X\overline{A}B^{\prime}Y})-\frac{\varepsilon_{0}%
}{1+\varepsilon_{0}}\operatorname{Tr}(\Delta_{X\overline{A}B^{\prime}Y}^{\prime
}\log\kappa_{X\overline{A}B^{\prime}Y})\\
&  =-h_{2}(\varepsilon_{0}/\left[  1+\varepsilon_{0}\right]  )+\frac
{1}{1+\varepsilon_{0}}D(\rho_{X\overline{A}B^{\prime}Y}\Vert\kappa_{X\overline
{A}B^{\prime}Y})\nonumber\\
&  \qquad+\frac{\varepsilon_{0}}{1+\varepsilon_{0}}D(\Delta_{X\overline{A}%
B^{\prime}Y}^{\prime}\Vert\kappa_{X\overline{A}B^{\prime}Y})\\
&  \geq-h_{2}(\varepsilon_{0}/\left[  1+\varepsilon_{0}\right]  )+\frac
{1}{1+\varepsilon_{0}}\inf_{\hat{\sigma}\in\operatorname{LHS}}D(\rho_{X\overline
{A}B^{\prime}Y}\Vert\sigma_{X\overline{A}B^{\prime}Y})\nonumber\\
&  \qquad+\frac{\varepsilon_{0}}{1+\varepsilon_{0}}\inf_{\hat{\omega}%
\in\operatorname{LHS}}D(\Delta_{X\overline{A}B^{\prime}Y}^{\prime}\Vert
\omega_{X\overline{A}B^{\prime}Y}).
\end{align}
The first inequality follows because%
\begin{equation}
H(\lambda\xi_{0}+(1-\lambda)\xi_{1})\leq H(\{\lambda,1-\lambda\})+\lambda
H(\xi_{0})+(1-\lambda)H(\xi_{1})
\end{equation}
for $\lambda\in\lbrack0,1]$ and density operators $\xi_{0}$ and $\xi_{1}$ and
where we define $h_{2}(\lambda):=H(\{\lambda,1-\lambda\})$. Since we have
shown that the above inequality holds for an arbitrary state $\kappa_{X\overline
{A}B^{\prime}Y}$ arising from an LHS\ assemblage $\hat{\kappa}$, we can
conclude that%
\begin{multline}
\inf_{\hat{\kappa}\in\operatorname{LHS}}D(\zeta_{X\overline{A}B^{\prime}Y}%
\Vert\kappa_{X\overline{A}B^{\prime}Y})\geq-h_{2}(\varepsilon_{0}/\left[
1+\varepsilon_{0}\right]  )\label{eq:cont-low-bound}\\
+\frac{1}{1+\varepsilon_{0}}\inf_{\hat{\sigma}\in\operatorname{LHS}}%
D(\rho_{X\overline{A}B^{\prime}Y}\Vert\sigma_{X\overline{A}B^{\prime}Y})+\frac
{\varepsilon_{0}}{1+\varepsilon_{0}}\inf_{\hat{\omega}\in\operatorname{LHS}%
}D(\Delta_{X\overline{A}B^{\prime}Y}^{\prime}\Vert\omega_{X\overline{A}B^{\prime}Y}).
\end{multline}
Putting the bounds in \eqref{eq:cont-up-bound} and
\eqref{eq:cont-low-bound}\ together and multiplying by $1+\varepsilon_{0}$, we
conclude that%
\begin{align}
&  \!\!\!\!\!\!\inf_{\hat{\sigma}\in\operatorname{LHS}}D(\rho_{X\overline
{A}B^{\prime}Y}\Vert\sigma_{X\overline{A}B^{\prime}Y})+\varepsilon_{0}\inf
_{\hat{\omega}\in\operatorname{LHS}}D(\Delta_{X\overline{A}B^{\prime}Y}^{\prime
}\Vert\omega_{X\overline{A}B^{\prime}Y})-g(\varepsilon_{0})\nonumber\\
&  \leq\inf_{\hat{\sigma}\in\operatorname{LHS}}D(\theta_{X\overline{A}B^{\prime}%
Y}\Vert\sigma_{X\overline{A}B^{\prime}Y})+\varepsilon_{0}\inf_{\hat{\omega}%
\in\operatorname{LHS}}D(\Delta_{X\overline{A}B^{\prime}Y}\Vert\omega_{X\overline
{A}B^{\prime}Y})\\
&  \leq\inf_{\hat{\sigma}\in\operatorname{LHS}}D(\theta_{X\overline{A}B^{\prime}%
Y}\Vert\sigma_{X\overline{A}B^{\prime}Y})+\varepsilon_{0}\log_{2}|\overline{A}|,
\end{align}
where we have used that $g(\varepsilon_{0})=\left(  1+\varepsilon_{0}\right)
h_{2}(\varepsilon_{0}/\left[  1+\varepsilon_{0}\right]  )$ \cite{Shirokov2016}
and Proposition~\ref{prop:rel-ent-st-up-bnd}. By dropping the term
$\varepsilon_{0}\inf_{\hat{\omega}\in\operatorname{LHS}}D(\Delta_{X\overline
{A}B^{\prime}Y}^{\prime}\Vert\omega_{X\overline{A}B^{\prime}Y})$ (it is
non-negative), we can rewrite the above bound as%
\begin{align}
\inf_{\hat{\sigma}\in\operatorname{LHS}}D(\rho_{X\overline{A}B^{\prime}Y}%
\Vert\sigma_{X\overline{A}B^{\prime}Y}) &  \leq\inf_{\hat{\sigma}\in
\operatorname{LHS}}D(\theta_{X\overline{A}B^{\prime}Y}\Vert\sigma_{X\overline
{A}B^{\prime}Y})+\varepsilon_{0}\log_{2}|\overline{A}|+g(\varepsilon_{0})\\
&  \leq\inf_{\hat{\sigma}\in\operatorname{LHS}}D(\theta_{X\overline{A}B^{\prime}%
Y}\Vert\sigma_{X\overline{A}B^{\prime}Y})+\varepsilon\log_{2}|\overline{A}%
|+g(\varepsilon),
\end{align}
where in the last line we have used the facts that $\varepsilon_{0}%
\leq\varepsilon$ and the function $\varepsilon\log_{2}|\overline{A}|+g(\varepsilon
)$ is monotone non-decreasing with respect to $\varepsilon$. Since the above
inequality holds for an arbitrary 1W-LOCC\ operation $\{p_{X|Y},\{\mathcal{K}%
_{y}\}_{y}\}$, we can conclude that%
\begin{multline}
\sup_{\{p_{X|Y},\{\mathcal{K}_{y}\}_{y}\}}\inf_{\hat{\sigma}\in
\operatorname{LHS}}D(\rho_{X\overline{A}B^{\prime}Y}\Vert\sigma_{X\overline{A}B^{\prime
}Y})\\
\leq\sup_{\{p_{X|Y},\{\mathcal{K}_{y}\}_{y}\}}\inf_{\hat{\sigma}%
\in\operatorname{LHS}}D(\theta_{X\overline{A}B^{\prime}Y}\Vert\sigma_{X\overline
{A}B^{\prime}Y})+\varepsilon\log_{2}|\overline{A}|+g(\varepsilon),
\end{multline}
which is the same as%
\begin{equation}
R_{S}(\overline{A};B)_{\hat{\rho}}\leq R_{S}(\overline{A};B)_{\hat{\theta}}%
+\varepsilon\log_{2}|\overline{A}|+g(\varepsilon).
\end{equation}
To get the other inequality $R_{S}(\overline{A};B)_{\hat{\theta}}\leq R_{S}(\overline
{A};B)_{\hat{\rho}}+\varepsilon\log_{2}|\overline{A}|+g(\varepsilon)$, we simply
repeat all of the above steps with $\hat{\rho}$ and $\hat{\theta}$ swapped.
\end{proof}

\subsection{Faithfulness\label{sec:faithful}}

A steering quantifier is \textit{faithful} if it is equal to zero if and only
if the assemblage has a local-hidden-state model. In this section, we prove
quantitative statements regarding the faithfulness of relative entropy of
steering. We begin with the implication $\hat{\rho}\in\operatorname{LHS}%
\Longrightarrow R_{S}(\overline{A};B)_{\hat{\rho}}=0$.

\begin{proposition}
Let $\varepsilon\in\left[  0,1\right]  $, and let $\{\hat{\rho}_{B}%
^{a,x}\}_{a,x}$ and $\{\hat{\sigma}_{B}^{a,x}\}_{a,x}$ be assemblages such
that $\hat{\sigma}\in\operatorname{LHS}$ and $\Delta(\hat{\rho},\hat{\sigma
})\leq\varepsilon$. Then%
\begin{equation}
R_{S}(\overline{A};B)_{\hat{\rho}}\leq\varepsilon\log_{2}|\overline{A}|+g(\varepsilon).
\end{equation}

\end{proposition}

\begin{proof}
This is a direct consequence of Proposition~\ref{thm:continuity}\ and the fact
that $\hat{\sigma}\in\operatorname{LHS}$ so that by applying
Definition~\ref{def:rel-ent-steer}, we see that $R_{S}(\overline{A};B)_{\hat
{\sigma}}=0$.
\end{proof}

\bigskip We now establish the implication $R_{S}(\overline{A};B)_{\hat{\rho}%
}=0\Longrightarrow\hat{\rho}\in\operatorname{LHS}$:

\begin{proposition}
Let $\{\hat{\rho}_{B}^{a,x}\}_{a,x}$ be an assemblage. Then%
\begin{equation}
\sqrt{2\ln2\ R_{S}(\overline{A};B)_{\hat{\rho}}}\geq\inf_{\{\hat{\sigma}_{B}%
^{a,x}\}_{a,x}\in\operatorname{LHS}}\frac{1}{\left\vert \mathcal{X}\right\vert
}\sum_{x,a}\left\Vert \hat{\rho}_{B}^{a,x}-\hat{\sigma}_{B}^{a,x}\right\Vert
_{1}.\label{eq:faithful-conv}%
\end{equation}
In particular, if $R_{S}(\overline{A};B)_{\hat{\rho}}=0$, then $\{\hat{\rho}%
_{B}^{a,x}\}_{a,x}\in\operatorname{LHS}$.
\end{proposition}

\begin{proof}
The inequality in \eqref{eq:faithful-conv} is a direct consequence of the
quantum Pinsker inequality \cite[Theorem~1.15]{OP93}, which is the statement
that%
\begin{equation}
D(\omega\Vert\tau)\geq\frac{1}{2\ln2}\left\Vert \omega-\tau\right\Vert
_{1}^{2},
\end{equation}
for quantum states $\omega$ and $\tau$. Applying it and definitions, we find
that%
\begin{align}
&  \!\!\!\!\!\!\sqrt{2\ln2\ R_{S}(\overline{A};B)_{\hat{\rho}}}\nonumber\\
&  \geq\sup_{\{p_{X|Y},\{\mathcal{K}_{y}\}_{y}\}}\inf_{\{\hat{\sigma}%
_{B}^{a,x}\}_{a,x}\in\operatorname{LHS}}\left\Vert \rho_{XAB^{\prime}Y}%
-\sigma_{XAB^{\prime}Y}\right\Vert _{1}\\
&  \geq\inf_{\{\hat{\sigma}_{B}^{a,x}\}_{a,x}\in\operatorname{LHS}}\left\Vert
\frac{1}{\left\vert \mathcal{X}\right\vert }\sum_{x,a}[x]\otimes\lbrack
a]\otimes\hat{\rho}_{B}^{a,x}-\frac{1}{\left\vert \mathcal{X}\right\vert }%
\sum_{x,a}[x]\otimes\lbrack a]\otimes\hat{\sigma}_{B}^{a,x}\right\Vert _{1}\\
&  =\inf_{\{\hat{\sigma}_{B}^{a,x}\}_{a,x}\in\operatorname{LHS}}\frac
{1}{\left\vert \mathcal{X}\right\vert }\sum_{x,a}\left\Vert \hat{\rho}%
_{B}^{a,x}-\hat{\sigma}_{B}^{a,x}\right\Vert _{1},
\end{align}
where the second inequality follows by picking a 1W-LOCC\ operation to be
trivial, consisting of choosing the input $x$ uniformly at random and applying
the identity channel to system$~B$.

To get the implication $R_{S}(\overline{A};B)_{\hat{\rho}}=0\Longrightarrow
\{\hat{\rho}_{B}^{a,x}\}_{a,x}\in\operatorname{LHS}$, consider that the trace
distance is continuous and the set $\operatorname{LHS}$\ is compact, so that
the infimum can be replaced with a minimum and thus in the case that
$R_{S}(\overline{A};B)_{\hat{\rho}}=0$, we can conclude that there exists
$\{\hat{\sigma}_{B}^{a,x}\}_{a,x}\in\operatorname{LHS}$ such that $\hat{\rho
}_{B}^{a,x}=\hat{\sigma}_{B}^{a,x}$ for all $a$ and $x$.
\end{proof}

\section{Restricted relative entropy of steering}

In this section, we define the restricted relative entropy of steering and establish several of its properties. As discussed in \cite{KWW16} and reviewed in the introduction, this quantity is motivated by the fact that a restricted class of
1W-LOCC\ operations might have more relevance in practical scenarios, in
which classical communication from Bob to Alice reaches Alice only after she
obtains the output of her black box. We begin by defining the restricted relative entropy of steering as follows:

\begin{definition}
[Restricted relative entropy of steering]Let $\{\hat{\rho}_{B}^{a,x}\}_{a,x}$
be an assemblage. Then the restricted relative entropy of steering is given by%
\begin{equation}
R_{S}^{R}(\overline{A};B)_{\hat{\rho}}:=\sup_{p_{X}}\inf_{\{\hat{\sigma}_{B}%
^{a,x}\}_{a,x}\in\operatorname{LHS}}D(\rho_{X\overline{A}B}\Vert\sigma_{X\overline{A}%
B}),
\end{equation}
where%
\begin{align}
\rho_{X\overline{A}B} &  :=\sum_{x,a}p_{X}(x)|x\rangle\langle x|_{X}\otimes
|a\rangle\langle a|_{\overline{A}}\otimes\hat{\rho}_{B}^{a,x},\\
\sigma_{X\overline{A}B} &  :=\sum_{x,a}p_{X}(x)|x\rangle\langle x|_{X}%
\otimes|a\rangle\langle a|_{\overline{A}}\otimes\hat{\sigma}_{B}^{a,x}.
\end{align}

\end{definition}

We first note that an exchange of the optimizations is possible for restricted
relative entropy of steering, due to its simpler form:

\begin{proposition}
Let $\{\hat{\rho}_{B}^{a,x}\}_{a,x}$ be an assemblage. Then%
\begin{equation}
R_{S}^{R}(\overline{A};B)_{\hat{\rho}}=\inf_{\{\hat{\sigma}_{B}^{a,x}\}_{a,x}%
\in\operatorname{LHS}}\sup_{p_{X}}D(\rho_{X\overline{A}B}\Vert\sigma_{X\overline{A}B}).
\end{equation}

\end{proposition}

\begin{proof}
We can use \eqref{eq:rel-ent-block-prop}\ to rewrite $D(\rho_{X\overline{A}B}%
\Vert\sigma_{X\overline{A}B})$ as follows:%
\begin{equation}
D(\rho_{X\overline{A}B}\Vert\sigma_{X\overline{A}B})=\sum_{x}p_{X}(x)D(\hat{\rho}%
_{\overline{A}B}^{x}\Vert\hat{\sigma}_{\overline{A}B}^{x}),
\end{equation}
where%
\begin{equation}
\rho_{\overline{A}B}^{x}:=\sum_{a}|a\rangle\langle a|_{\overline{A}}\otimes\hat{\rho
}_{B}^{a,x},\qquad\sigma_{\overline{A}B}^{x}:=\sum_{a}|a\rangle\langle a|_{\overline{A}%
}\otimes\hat{\sigma}_{B}^{a,x}.
\end{equation}
After doing so, we see that the function $D(\rho_{X\overline{A}B}\Vert\sigma
_{X\overline{A}B})$ being optimized is linear in $p_{X}$ and convex in $\hat
{\sigma}_{B}^{a,x}$, the latter due to the well known joint convexity of
relative entropy (see, e.g., \cite{W15book}).\ So the Sion minimax theorem
\cite{S58}\ applies and allows for an exchange of the optimizations.
\end{proof}

The restricted relative entropy of steering obeys many properties similar to
those of the relative entropy of steering, and we mostly list them below
without proof because their proofs follow quite similarly to what we have
shown previously (i.e., in some cases, a proof seems necessary and so we give
it, while in others, a proof is an immediate consequence of prior developments
and so we do not give it).

The first is the following:

\begin{theorem}
[Restricted 1W-LOCC monotone]Let $\{\hat{\rho}_{B}^{a,x}\}_{a,x}$ be an
assemblage, and let
\begin{equation}
\{p_{X|X_{f}},p_{\overline{A}_{f}|\overline{A}XX_{f}Z},\{\mathcal{K}_{z}\}_{z}\}
\end{equation}
denote a restricted 1W-LOCC\ operation that results in an assemblage
$\{\hat{\omega}_{B^{\prime}}^{a_{f},x_{f}}\}_{a_{f},x_{f}}$, defined as%
\begin{equation}
\hat{\omega}_{B^{\prime}}^{a_{f},x_{f}}:=\sum_{a,x,z}p_{X|X_{f}}%
(x|x_{f})p_{\overline{A}_{f}|\overline{A}XX_{f}Z}(a_{f}|a,x,x_{f},z)\mathcal{K}_{z}%
(\hat{\rho}_{B}^{a,x}).
\end{equation}
Then%
\begin{equation}
R_{S}^{R}(\overline{A};B)_{\hat{\rho}}\geq R_{S}^{R}(\overline{A}_{f};B^{\prime}%
)_{\hat{\omega}}.
\end{equation}

\end{theorem}

\begin{proof}
Taking a distribution $p_{X_{f}}$ over the black-box inputs of the final
assemblage, we can embed the state of the final assemblage into the following
classical--quantum state:%
\begin{equation}
\omega_{X_{f}\overline{A}_{f}B^{\prime}}:=\sum_{x_{f},a_{f}}p_{X_{f}}(x_{f}%
)[x_{f}]\otimes\lbrack a_{f}]\otimes\hat{\omega}_{B^{\prime}}^{a_{f},x_{f}%
},\label{eq:cq-final-state-RRS}%
\end{equation}
which is a marginal of the following state:%
\begin{multline}
\omega_{X_{f}X\overline{A}_{f}\overline{A}ZB^{\prime}}:=\sum_{x_{f},a_{f},a,x,z}%
p_{X_{f}}(x_{f})[x_{f}]\otimes p_{X|X_{f}}(x|x_{f})[x]\\
\otimes p_{\overline{A}_{f}|\overline{A}XX_{f}Z}(a_{f}|a,x,x_{f},z)[a_{f}]\otimes\lbrack
a]\otimes\lbrack z]\otimes\mathcal{K}_{z}(\hat{\rho}_{B}^{a,x}).
\end{multline}
Let $\hat{\sigma}_{B^{\prime}}^{a_{f},x_{f}}$ denote an arbitrary LHS
assemblage, and let $\sigma_{X_{f}\overline{A}_{f}B^{\prime}}$ denote its
corresponding classical--quantum state:%
\begin{equation}
\sigma_{X_{f}\overline{A}_{f}B^{\prime}}:=\sum_{x_{f},a_{f}}p_{X_{f}}(x_{f}%
)[x_{f}]\otimes\lbrack a_{f}]\otimes\hat{\sigma}_{B^{\prime}}^{a_{f},x_{f}}.
\end{equation}
Let $\hat{\tau}_{B}^{a,x}$ denote an arbitrary LHS assemblage, and let
$\tau_{X_{f}X\overline{A}_{f}\overline{A}ZB^{\prime}}$ denote the following state:%
\begin{multline}
\tau_{X_{f}X\overline{A}_{f}\overline{A}ZB^{\prime}}:=\sum_{x_{f},a_{f},a,x,z}p_{X_{f}%
}(x_{f})[x_{f}]\otimes p_{X|X_{f}}(x|x_{f})[x]\\
\otimes p_{\overline{A}_{f}|\overline{A}XX_{f}Z}(a_{f}|a,x,x_{f},z)[a_{f}]\otimes\lbrack
a]\otimes\lbrack z]\otimes\mathcal{K}_{z}(\hat{\tau}_{B}^{a,x}).
\end{multline}
Define the following states:%
\begin{align}
\omega_{X_{f}X\overline{A}B}^{\prime} &  :=\sum_{x_{f},a_{f},a,x,z}p_{X_{f}}%
(x_{f})[x_{f}]\otimes p_{X|X_{f}}(x|x_{f})[x]\otimes\lbrack a]\otimes\hat
{\rho}_{B}^{a,x},\\
\tau_{X_{f}X\overline{A}B}^{\prime} &  :=\sum_{x_{f},a_{f},a,x,z}p_{X_{f}}%
(x_{f})[x_{f}]\otimes p_{X|X_{f}}(x|x_{f})[x]\otimes\lbrack a]\otimes\hat
{\tau}_{B}^{a,x}.
\end{align}
Consider that $\tau_{X_{f}\overline{A}_{f}B^{\prime}}$ corresponds to an
LHS\ assemblage by \cite[Theorem~1]{Gallego2015}. Then%
\begin{align}
\inf_{\{\hat{\sigma}_{B^{\prime}}^{a_{f},x_{f}}\}_{a_{f},x_{f}}\in
\operatorname{LHS}}D(\omega_{X_{f}\overline{A}_{f}B^{\prime}}\Vert\sigma_{X_{f}%
\overline{A}_{f}B^{\prime}}) &  \leq D(\omega_{X_{f}\overline{A}_{f}B^{\prime}}\Vert
\tau_{X_{f}\overline{A}_{f}B^{\prime}})\\
&  \leq D(\omega_{X_{f}X\overline{A}_{f}\overline{A}ZB^{\prime}}\Vert\tau_{X_{f}X\overline
{A}_{f}\overline{A}ZB^{\prime}})\\
&  =D(\omega_{X_{f}X\overline{A}ZB^{\prime}}\Vert\tau_{X_{f}X\overline{A}ZB^{\prime}})\\
&  \leq D(\omega_{X_{f}X\overline{A}B}^{\prime}\Vert\tau_{X_{f}X\overline{A}B}^{\prime
})\\
&  =D(\omega_{X\overline{A}B}^{\prime}\Vert\tau_{X\overline{A}B}^{\prime}).
\end{align}
The first inequality follows because $\tau_{X_{f}\overline{A}_{f}B^{\prime}}$
corresponds to a particular LHS\ assemblage. The second inequality follows
from the data-processing inequality. The equality follows due to
\begin{equation}
D(\omega_{X_{f}X\overline{A}_{f}\overline{A}ZB^{\prime}}\Vert\tau_{X_{f}X\overline{A}_{f}%
\overline{A}ZB^{\prime}})\leq D(\omega_{X_{f}X\overline{A}ZB^{\prime}}\Vert\tau
_{X_{f}X\overline{A}ZB^{\prime}}),
\end{equation}
which is a consequence of the fact that register $\overline{A}_{f}$ results from
processing the values in $X_{f}X\overline{A}Z$ according to $p_{\overline{A}_{f}|\overline
{A}XX_{f}Z}$, while the opposite inequality
\begin{equation}
D(\omega_{X_{f}X\overline{A}_{f}\overline{A}ZB^{\prime}}\Vert\tau_{X_{f}X\overline{A}_{f}%
\overline{A}ZB^{\prime}})\geq D(\omega_{X_{f}X\overline{A}ZB^{\prime}}\Vert\tau
_{X_{f}X\overline{A}ZB^{\prime}})
\end{equation}
follows because partial trace over $\overline{A}_{f}$ is a channel. The final
inequality again follows from data processing:\ We get $\omega_{X_{f}X\overline
{A}ZB^{\prime}}$ from $\omega_{X_{f}X\overline{A}B}^{\prime}$ and $\tau_{X_{f}%
X\overline{A}ZB^{\prime}}$ from $\tau_{X_{f}X\overline{A}B}^{\prime}$ by performing the
quantum channel $(\cdot)\rightarrow\sum_{z}[z]\otimes\mathcal{K}_{z}(\cdot)$
on system $B$. The final equality follows again from data
processing:\ $D(\omega_{X_{f}X\overline{A}B}^{\prime}\Vert\tau_{X_{f}X\overline{A}%
B}^{\prime})\geq D(\omega_{X\overline{A}B}^{\prime}\Vert\tau_{X\overline{A}B}^{\prime})$
because partial trace is a channel and $D(\omega_{X_{f}X\overline{A}B}^{\prime
}\Vert\tau_{X_{f}X\overline{A}B}^{\prime})\leq D(\omega_{X\overline{A}B}^{\prime}%
\Vert\tau_{X\overline{A}B}^{\prime})$ because we can apply the Bayes theorem to see
that $p_{X_{f}}p_{X|X_{f}}=p_{X_{f}|X}p_{X}$ and thus $X_{f}$ can be see to
arise from processing of $X$. Since we have shown that the inequality holds
for an arbitrary LHS\ assemblage, we can conclude that%
\begin{align}
\inf_{\{\hat{\sigma}_{B^{\prime}}^{a_{f},x_{f}}\}_{a_{f},x_{f}}\in
\operatorname{LHS}}D(\omega_{X_{f}\overline{A}_{f}B^{\prime}}\Vert\sigma_{X_{f}%
\overline{A}_{f}B^{\prime}}) &  \leq\inf_{\{\hat{\tau}_{B}^{a,x}\}_{a,x}%
\in\operatorname{LHS}}D(\omega_{X\overline{A}B}^{\prime}\Vert\tau_{X\overline{A}%
B}^{\prime})\\
&  \leq\sup_{p_{X}}\inf_{\{\hat{\tau}_{B}^{a,x}\}_{a,x}\in\operatorname{LHS}%
}D(\omega_{X\overline{A}B}^{\prime}\Vert\tau_{X\overline{A}B}^{\prime})\\
&  =R_{S}^{R}(\overline{A};B)_{\hat{\rho}}.
\end{align}
Since the above holds for an arbitrary distribution $p_{X_{f}}$, we can
conclude that%
\begin{equation}
\sup_{p_{X_{f}}}\inf_{\{\hat{\sigma}_{B^{\prime}}^{a_{f},x_{f}}\}_{a_{f}%
,x_{f}}\in\operatorname{LHS}}D(\omega_{X_{f}\overline{A}_{f}B^{\prime}}\Vert
\sigma_{X_{f}\overline{A}_{f}B^{\prime}})\leq R_{S}^{R}(\overline{A};B)_{\hat{\rho}},
\end{equation}
which is equivalent to the statement of the theorem.
\end{proof}

\begin{proposition}
[Convexity]Let $\lambda\in\lbrack0,1]$. Let $\{\hat{\rho}_{B}^{a,x}\}_{a,x}$
and $\{\hat{\theta}_{B}^{a,x}\}_{a,x}$ be two assemblages, and consider an
assemblage $\{\hat{\tau}_{B}^{a,x}:=\lambda\hat{\rho}_{B}^{a,x}+(1-\lambda
)\hat{\theta}_{B}^{a,x}\}_{a,x}$. The restricted relative entropy of steering
is convex in the following sense:%
\begin{equation}
R_{S}^{R}(\overline{A};B)_{\hat{\tau}}\leq\lambda R_{S}^{R}(\overline{A};B)_{\hat{\rho}%
}+(1-\lambda)R_{S}^{R}(\overline{A};B)_{\hat{\theta}}.
\end{equation}

\end{proposition}

The next proposition finds several upper bounds on the restricted relative
entropy of steering, one of which is in terms of the conditional mutual
information, defined for a tripartite state $\varsigma_{KLM}$ as
$I(K;L|M)_{\varsigma}:=H(KM)_{\varsigma}+H(LM)_{\varsigma}-H(M)_{\varsigma
}-H(KLM)_{\varsigma}$.

\begin{proposition}
[Upper bounds]\label{prop:rel-ent-st-up-bnd copy(1)}Let $\{\hat{\rho}%
_{B}^{a,x}\}_{a,x}$ be an assemblage. Then%
\begin{align}
R_{S}^{R}(\overline{A};B)_{\hat{\rho}} &  \leq\sup_{p_{X}}I(\overline{A};B|X)_{\rho
}\label{eq:RRS-upper-bnd-1}\\
&  \leq\min\left\{  \sup_{p_{X}}H(\overline{A}),H(B)_{\rho}\right\}
\label{eq:RRS-upper-bnd-2}\\
&  \leq\min\left\{  \log_{2}|\overline{A}|,\log_{2}\left\vert B\right\vert
\right\}  ,\label{eq:RRS-upper-bnd-3}%
\end{align}
where the conditional mutual information is with respect to the following
state:%
\begin{equation}
\rho_{X\overline{A}B^{\prime}Y}:=\sum_{x,a}p_{X}(x)[x]\otimes\lbrack a]\otimes
\hat{\rho}_{B}^{a,x},
\end{equation}
and the entropy $H(\overline{A})$\ is with respect to the probability distribution
$p_{\overline{A}}(a):=\sum_{x}p_{X}(x)\operatorname{Tr}(\hat{\rho}_{B}^{a,x})$.
\end{proposition}

\begin{proof}
We can choose an assemblage having a local-hidden state model to be as
follows:%
\begin{equation}
\{\hat{\xi}_{B}^{a,x}:=p_{\overline{A}|X}(a|x)\rho_{B}\}_{a,x},
\end{equation}
where $p_{\overline{A}|X}(a|x)=\operatorname{Tr}(\hat{\rho}_{B}^{a,x})$ and
$\rho_{B}=\sum_{a}\hat{\rho}_{B}^{a,x}$ (recall the no-signaling condition in
\eqref{eq:no-sig-constr}). Then define the following state:%
\begin{equation}
\xi_{X\overline{A}B}:=\sum_{x,a}p_{X}(x)[x]\otimes\lbrack a]\otimes\hat{\xi}%
_{B}^{a,x}=\left[  \sum_{x,a}p_{X}(x)p_{\overline{A}|X}(a|x)[x]\otimes\lbrack
a]\right]  \otimes\rho_{B}=\rho_{X\overline{A}}\otimes\rho_{B}.
\end{equation}
Consider that%
\begin{align}
\inf_{\{\hat{\sigma}_{B}^{a,x}\}_{a,x}\in\operatorname{LHS}}D(\rho_{X\overline{A}%
B}\Vert\sigma_{X\overline{A}B}) &  \leq D(\rho_{X\overline{A}B}\Vert\xi_{X\overline{A}B})\\
&  =I(X\overline{A};B)_{\rho}\\
&  =I(\overline{A};B|X)_{\rho}+I(X;B)_{\rho}\\
&  =I(\overline{A};B|X)_{\rho}.
\end{align}
The inequality follows because the state $\xi_{X\overline{A}B}$ arises from a
particular LHS\ assemblage. The first equality follows from the well known
characterization of mutual information as the relative entropy of the joint
state to the product of the marginals. The second equality follows from the
chain rule for conditional mutual information (see, e.g., \cite{W15book}), and
the last from the no-signaling condition in \eqref{eq:no-sig-constr}, so that
$I(X;B)_{\rho}=0$. Since the inequality holds for all distributions $p_{X}$,
we can take a supremum to arrive at \eqref{eq:RRS-upper-bnd-1}. The latter two
inequalities in \eqref{eq:RRS-upper-bnd-2} and \eqref{eq:RRS-upper-bnd-3}
follow from well known bounds on conditional mutual information (see, e.g.,
\cite{W15book}), and using that systems $\overline{A}$ and $X$ are classical.
\end{proof}

\begin{definition}
[Restricted trace distance of assemblages]Let $\{\hat{\rho}_{B}^{a,x}\}_{a,x}$
and $\{\hat{\theta}_{B}^{a,x}\}_{a,x}$ be two assemblages. We define the
restricted normalized trace distance of assemblages as%
\begin{equation}
\Delta^{R}(\hat{\rho},\hat{\theta}):=\frac{1}{2}\sup_{p_{X}}\left\Vert
\rho_{X\overline{A}B}-\theta_{X\overline{A}B}\right\Vert _{1},
\end{equation}
where%
\begin{align}
\rho_{X\overline{A}B} &  :=\sum_{x,a}p_{X}(x)[x]\otimes\lbrack a]\otimes\hat{\rho
}_{B}^{a,x},\\
\theta_{X\overline{A}B} &  :=\sum_{x,a}p_{X}(x)[x]\otimes\lbrack a]\otimes
\hat{\theta}_{B}^{a,x}.
\end{align}

\end{definition}

\begin{proposition}
[Metric]The restricted trace distance of assemblages is a metric.
\end{proposition}

\begin{theorem}
[Uniform continuity]\label{thm:continuity copy(1)}Let $\{\hat{\rho}_{B}%
^{a,x}\}_{a,x}$ and $\{\hat{\theta}_{B}^{a,x}\}_{a,x}$ be two assemblages such
that $\Delta^{R}(\hat{\rho},\hat{\theta})\leq\varepsilon\in\left[  0,1\right]
$. Then%
\begin{equation}
\left\vert R_{S}^{R}(\overline{A};B)_{\hat{\rho}}-R_{S}^{R}(\overline{A};B)_{\hat
{\theta}}\right\vert \leq\varepsilon\log_{2}\min\{|\overline{A}%
|,|B|\}+g(\varepsilon),
\end{equation}
where $g(\varepsilon):=(\varepsilon+1)\log_{2}(\varepsilon+1)-\varepsilon
\log_{2}\varepsilon$.
\end{theorem}

\begin{proposition}
Let $\varepsilon\in\left[  0,1\right]  $, and let $\{\hat{\rho}_{B}%
^{a,x}\}_{a,x}$ and $\{\hat{\sigma}_{B}^{a,x}\}_{a,x}$ be assemblages such
that $\hat{\sigma}\in\operatorname{LHS}$ and $\Delta^{R}(\hat{\rho}%
,\hat{\sigma})\leq\varepsilon$. Then%
\begin{equation}
R_{S}^{R}(\overline{A};B)_{\hat{\rho}}\leq\varepsilon\log_{2}\min\{|\overline
{A}|,|B|\}+g(\varepsilon).
\end{equation}

\end{proposition}

\begin{proposition}
Let $\{\hat{\rho}_{B}^{a,x}\}_{a,x}$ be an assemblage. Then%
\begin{equation}
\sqrt{2\ln2\ R_{S}^{R}(\overline{A};B)_{\hat{\rho}}}\geq\inf_{\{\hat{\sigma}%
_{B}^{a,x}\}_{a,x}\in\operatorname{LHS}}\frac{1}{\left\vert \mathcal{X}%
\right\vert }\sum_{x,a}\left\Vert \hat{\rho}_{B}^{a,x}-\hat{\sigma}_{B}%
^{a,x}\right\Vert _{1}.
\end{equation}
In particular, if $R_{S}^{R}(\overline{A};B)_{\hat{\rho}}=0$, then $\{\hat{\rho
}_{B}^{a,x}\}_{a,x}\in\operatorname{LHS}$.
\end{proposition}

\section{Conclusion}

We provided a definition of relative entropy of steering different from that
in \cite{Gallego2015}, justifying it based on well grounded
information-theoretic and game-theoretic concerns. We showed how this modified
relative entropy of steering satisfies several desirable properties, including
convexity, steering monotonicity, continuity, and faithfulness. We also
considered a restricted relative entropy of steering, which is relevant as a
quantifier in an operational setting in which there might be further
restrictions on one-way local operations and classical communication, as
discussed previously in \cite{KWW16}. The restricted relative entropy of
steering is also convex, a steering monotone, continuous, and faithful. Going
forward, we suspect that the definitions proposed here will be relevant in
applications of steering, such as one-sided device-independent quantum key
distribution, but we leave this for future work.

\bigskip\textbf{Acknowledgements.} We are grateful to Rodrigo Gallego for
discussions related to the topic of this paper. EK acknowledges support from
the Department of Physics and Astronomy at LSU. MMW acknowledges support from
the NSF under Award No.~CCF-1350397.

\bibliographystyle{alpha}
\bibliography{steering}

\end{document}